%% file: main.tex
\documentclass[12pt,letterpaper,onecolumn,draftclsnofoot]{IEEEtran}

\usepackage[letterpaper,margin=1in]{geometry}
\input{macros}

\title{Lower bounds for Bayesian CVaR in interactive problems with local differential privacy}
\author{%
Raghav Bongole, Amirreza Zamani, Tobias J. Oechtering,
and Mikael Skoglund%
\thanks{R. Bongole, A. Zamani, T. J. Oechtering, and M. Skoglund
are with the Department of Information Science and Engineering,
School of Electrical Engineering and Computer Science,
KTH Royal Institute of Technology,
Sweden
(e-mail: \{bongole,amizam,oech,skoglund\}@kth.se).}%
\thanks{Preliminary versions of parts of this work appeared in
\cite{bongole2026generalized,bongole2026instantiating}.}%
}

\begin{document}
\maketitle

\begin{abstract}
We develop an information-theoretic method for lower bounding the prior-predictive conditional value-at-risk (CVaR) of loss in interactive statistical decision making (ISDM). A generalized Fano inequality replaces a single error event by a randomized binary statistic constructed from a bounded transform of the loss, reducing the problem to a Bernoulli $f$-divergence inequality. A bounded hinge transform yields a direct lower bound on Bayesian CVaR under a prior fixed in advance. We derive  reverse-KL specializations and obtain explicit bounds for Gaussian mean estimation and Gaussian stochastic bandits. We then impose local differential privacy (LDP) at the observation stage of the same ISDM model. LDP contraction inequalities combined with the adaptive KL chain rule give explicit Bayesian CVaR lower bounds for Gaussian mean estimation with heterogeneous privacy parameters and Gaussian $K$-armed bandits. 
\end{abstract}


\input{sections/01_introduction}
\input{sections/02_isdm}

\input{sections/03_generalized_fano}

\input{sections/04_bayesian_cvar}
\input{sections/05_method}
\input{sections/06_gaussian}

\input{sections/07_ldp_isdm}
\input{sections/08_ldp_applications}
\input{sections/09_related_work}
\input{sections/10_discussion}
\input{sections/11_conclusion}

\bibliographystyle{IEEEtran}
\bibliography{references}

\newpage
\appendices
\input{appendices/A_generalized_fano}
\input{appendices/B_divergence}
\input{appendices/C_gaussian}
\input{appendices/D_ldp}

\end{document}

%% file: macros.tex
\usepackage{amsmath,amssymb,amsthm,mathtools,bm,bbm,microtype,hyperref,cleveref,enumitem,booktabs,tabularx,array,float}
\hypersetup{hidelinks,hypertexnames=false}
\newcommand{\Alg}{\mathsf{ALG}}
\newcommand{\Dcal}{\mathcal{D}}
\newcommand{\Mcal}{\mathcal{M}}

\newcommand{\Ycal}{\mathcal{Y}}
\newcommand{\Zcal}{\mathcal{Z}}
\newcommand{\Acal}{\mathcal{A}}
\newcommand{\Hcal}{\mathcal{H}}
\newcommand{\epsvec}{\boldsymbol{\varepsilon}}
\newcommand{\CVaR}{\operatorname{CVaR}}
\newcommand{\VaR}{\operatorname{VaR}}
\newcommand{\KL}{D_{\mathrm{KL}}}
\newcommand{\TV}{d_{\mathrm{TV}}}
\newcommand{\Bern}{\operatorname{Bern}}
\newcommand{\Law}{\operatorname{Law}}
\newcommand{\R}{\mathbb{R}}
\newcommand{\E}{\mathbb{E}}

\newcommand{\1}{\mathbbm{1}}
\newcommand{\defeq}{\mathrel{:=}}
\newcommand{\Cval}{\mathfrak C}

\newtheorem{theorem}{Theorem}[section]
\newtheorem{proposition}[theorem]{Proposition}
\newtheorem{lemma}[theorem]{Lemma}
\newtheorem{corollary}[theorem]{Corollary}
\newtheorem{definition}[theorem]{Definition}
\newtheorem{remark}[theorem]{Remark}

%% file: sections/01_introduction.tex
\section{Introduction}
\label{sec:introduction}

Interactive statistical decision making (ISDM) provides a common
framework for problems in which a learner's actions influence the
information it subsequently observes. A model specifies the
environment, an algorithm selects actions using the available history,
and their interaction induces a distribution over the resulting
transcript. Passive estimation is recovered when the learner has no
control over the observation process, whereas multi-armed bandits,
contextual bandits, and reinforcement learning are prominent adaptive
examples
\cite{foster2021statistical,foster2023tight,chen2024assouad}.

Information-theoretic lower bounds characterize the performance that
no admissible procedure can surpass. Classical methods based on
Le Cam's, Assouad's, and Fano's inequalities relate statistical
difficulty to the distinguishability of candidate models. In
interactive problems, this distinguishability depends on the
algorithm, because its actions determine which observations are
collected. Interactive extensions of classical lower-bound methods
account for this dependence
\cite{chen2024assouad}, while the Decision--Estimation Coefficient
provides a complementary characterization of the complexity of broad
classes of interactive decision problems
\cite{foster2021statistical,foster2023tight}.

Most information-theoretic lower bounds focus on expected loss or on
the probability that a loss exceeds a prescribed threshold.
Interactive Fano arguments, for example, typically reduce the loss to
a binary success or failure event and then control the corresponding
Bernoulli divergence
\cite{chen2024assouad,bongole2026quantile}. This reduction is well
suited to expected-risk and quantile criteria, but it does not directly
describe more general bounded functionals of the loss.

The distinction is particularly relevant for conditional
value-at-risk (CVaR), which measures the severity of losses in the
upper tail rather than only their expectation or a single quantile.
Information-theoretic lower bounds for Bayesian expected risk are
well developed
\cite{esposito2024bayesianrisk}, and recent work has established
high-probability and minimax-quantile lower bounds in both passive and
interactive settings
\cite{ma2024highprobability,chen2024assouad,
bongole2026quantile,esposito2026minimax}. These results concern risk
criteria that differ from the prior-predictive Bayesian CVaR studied
here. In our formulation, a model is first drawn from a prescribed
prior, the transcript is then generated by its interaction with the
algorithm, and CVaR is evaluated under the resulting joint
prior-predictive law. The prior is fixed before optimizing over
admissible algorithms and therefore forms part of the decision
problem.

This paper develops a general route to lower bounds for this
criterion. The key step replaces the usual indicator of a threshold
event by a randomized binary statistic representing an arbitrary
bounded transform of the loss. Data processing then relates the
expected transformed loss under the model-dependent experiment to its
counterpart under a reference experiment through a Bernoulli
$f$-divergence constraint. Indicator functions recover the usual
threshold-event reduction, while other transforms permit different
loss functionals to be treated within the same argument.

For bounded losses, we apply this reduction to a hinge transform and
combine it with the Rockafellar--Uryasev representation of CVaR
\cite{rockafellar2002conditional}. The resulting converse reduces a
problem-specific CVaR lower bound to two tasks: controlling the
excess loss under a suitable reference distribution and bounding the
average divergence between the model-dependent and reference
transcript laws. A reverse-KL specialization produces an exponential
Bernoulli inverse. Balancing this inverse against the CVaR tail level
leads to the logarithmic dependence on the reciprocal tail probability
that appears in our explicit lower bounds.

We first apply the method without privacy constraints. For Gaussian
mean estimation, we use a symmetric two-point prior and a clipped
squared-error proxy to lower bound the prior-predictive CVaR of squared
estimation error. For Gaussian $K$-armed stochastic bandits, we use a
uniform prior over environments with different optimal arms and lower
bound the CVaR of cumulative pseudo-regret. In each case, the
fixed-prior inequality is established for every value of the model
separation. A separation is selected afterward to identify a
hard prior from the corresponding family. Thus, for each selected
prior, the optimization over algorithms is a fixed-prior
Bayesian decision problem.

We then incorporate local differential privacy (LDP). Statistical
estimation under LDP has an extensive information-theoretic literature
\cite{6686179,duchi2018minimax,duchi2019lower}, and interactive
decision making under LDP has also been studied through private
variants of the Decision--Estimation Coefficient and related
high-probability criteria
\cite{chen2025changing}. Heterogeneous privacy levels have recently
been considered for high-probability estimation
\cite{aliakbarpour2025highprobability}. Our objective is different:
we seek fixed-prior, prior-predictive CVaR lower bounds when the learner
acts using locally privatized observations.

The general CVaR argument remains unchanged in the private setting and
is applied directly to the released transcript. Privacy affects the
analysis through the available divergence budget. We conditionally
apply an LDP contraction inequality at each round and combine the
result with the adaptive KL chain rule. This yields transcript-level
divergence bounds for sequentially private procedures, which can then
be inserted into the same Bayesian CVaR converse used in the
nonprivate setting.

We instantiate this argument for Gaussian mean estimation with
heterogeneous privacy parameters and for Gaussian $K$-armed bandits
under sequential LDP. We also provide a bounded mean-estimation
example under heterogeneous LDP to compare the parameter dependence
of our Bayesian CVaR lower bound with known high-probability minimax
results for the same bounded observation model
\cite{aliakbarpour2025highprobability}. The comparison concerns the
dependence on the tail and privacy parameters and the Bayesian CVaR and
minimax-quantile criteria remain distinct.

\subsection{Contributions}
\label{subsec:contributions}

The main contributions are as follows.

\begin{enumerate}[leftmargin=*,itemsep=4pt]

\item
\textbf{A bounded-transform extension of interactive Fano arguments.}
We extend the interactive Fano method of
\cite{chen2024assouad} from threshold indicators to arbitrary bounded
measurable transforms of the loss. A randomized binary representation
and data processing yield a two-sided Bernoulli $f$-divergence
constraint on the expected transform. This allows a broader class of
loss functionals to be controlled directly, without first reducing the
loss to a single threshold event. The threshold-event reduction of \cite{chen2024assouad}
is recovered as a special case, while other choices, such as the hinge
transform, provide direct access to tail-sensitive criteria including
CVaR.

\item
\textbf{A fixed-prior Bayesian CVaR converse.}
For bounded losses, a hinge transform and the
Rockafellar--Uryasev representation convert the general
$f$-divergence constraint into a lower bound on prior-predictive
Bayesian CVaR. We develop an explicit reverse-KL specialization that
retains the dependence on the CVaR level.

\item
\textbf{Explicit Gaussian estimation and bandit lower bounds.}
We apply the converse to Gaussian mean estimation under squared error
and to Gaussian $K$-armed bandits under cumulative pseudo-regret.
The results first hold for every fixed prior in a parameterized hard
family. Selecting a separation afterward gives hard-prior lower bounds
with logarithmic dependence on the CVaR tail probability.

\item
\textbf{Extension to locally private observations.}
We formulate the released-transcript version of ISDM under sequential
LDP and combine one-round LDP contraction with the adaptive KL chain
rule. This gives Bayesian CVaR lower bounds for Gaussian mean
estimation with heterogeneous privacy levels, Gaussian stochastic
bandits under LDP, and a bounded mean-estimation model used for
comparison with existing heterogeneous-LDP results.

\end{enumerate}

Section~\ref{sec:generalized-fano} presents the bounded-transform
divergence inequality, and Section~\ref{sec:bayesian-cvar} derives its
Bayesian CVaR consequence. Section~\ref{sec:gaussian} treats the
Gaussian applications. Sections~\ref{sec:private-isdm}
and~\ref{sec:private-applications} develop the LDP formulation and its
applications. Related work and limitations are discussed in
Sections~\ref{sec:related-work} and~\ref{sec:discussion}, respectively,
and detailed proofs are deferred to the appendices.

Preliminary versions of parts of this work appeared in
\cite{bongole2026generalized,bongole2026instantiating}. The present
article combines the bounded-transform and Bayesian CVaR framework
with reverse-KL Gaussian applications and an extension to sequential
local differential privacy.

%% file: sections/02_isdm.tex
\section{Interactive Statistical Decision Making and Bayesian Risk}
\label{sec:isdm}
This section introduces the statistical framework and the Bayesian risk
criterion used throughout the paper. We first fix the notation, then
describe the interactive experiment and its induced transcript
distribution, and finally define the Bayesian CVaR criterion under a
prior fixed in advance.
\subsection{Notation}

We use capital letters for random variables and lower-case letters for their realizations. For probability measures $P$ and $Q$ on a common measurable space, total variation is denoted by $\TV(P,Q)$ and Kullback--Leibler divergence by $\KL(P\|Q)$. 
For probability measures $P$ and $Q$ on a common measurable space and
a convex function $f:\R_+\to\R\cup\{+\infty\}$ satisfying $f(1)=0$,
the $f$-divergence is
\begin{equation}
    D_f(P\|Q)
    \defeq
    \int q\,f\!\left(\frac{p}{q}\right)d\lambda,
\end{equation}
where $\lambda$ is any measure dominating $P$ and $Q$, with
$p=dP/d\lambda$ and $q=dQ/d\lambda$ \cite{polyanskiy2025information}. This definition is independent of the choice of
$\lambda$.

We write
\begin{equation}
    [x]_+\defeq\max\{x,0\}.
\end{equation}
The Bernoulli distribution with mean $a$ is denoted by $\Bern(a)$.

For a nonnegative random variable $L$ and $\alpha\in(0,1)$, the CVaR is written as
\begin{equation}
    \CVaR_\alpha(L)
    \defeq
    \inf_{s\in\R}
    \left\{
        s+
        \frac{1}{1-\alpha}\E[(L-s)_+]
    \right\}.
    \label{eq:isdm-cvar}
\end{equation}

\subsection{Interactive statistical decision making}

Let $\Mcal$ be a model class, $\Acal$ an action space, and $\Ycal$ an observation space. An interactive procedure $\Alg\in\Dcal$ operates for $T$ rounds. Given the history
\begin{equation}
    H_{t-1}=(A_1,Y_1,\ldots,A_{t-1},Y_{t-1}),
\end{equation}
it selects
\begin{equation}
    A_t\sim\pi_t(\cdot\mid H_{t-1}).
\end{equation}
Under a model $M\in\Mcal$, the environment then generates
\begin{equation}
    Y_t\sim P_t^M(\cdot\mid A_t,H_{t-1}).
\end{equation}
The transcript is
\begin{equation}
    H_T=(A_1,Y_1,\ldots,A_T,Y_T),
\end{equation}
and its induced distribution is denoted by
\begin{equation}
    \mathbb P^{M,\Alg}.
\end{equation}

The defining feature of ISDM is that the transcript distribution depends on the procedure. In bandit problems, the action determines which observation distribution is sampled. The same formulation subsumes both passive and adaptive statistical problems.

\subsection{Fixed-prior Bayesian decision problems}

Fix a prior
\begin{equation}
    \mu\in\mathcal P(\Mcal).
\end{equation}
 For a fixed procedure $\Alg$, the Bayesian experiment is
\begin{equation}
    M\sim\mu,
    \qquad
    H_T\mid M\sim\mathbb P^{M,\Alg}.
    \label{eq:bayesian-experiment}
\end{equation}
Equivalently, the joint law of the model and transcript is
\begin{equation}
    \mathbb P_{\mu,\Alg}(dM,dh)
    \defeq
    \mu(dM)\mathbb P^{M,\Alg}(dh).
    \label{eq:joint-law}
\end{equation}

Let
\begin{equation}
    L:\Mcal\times\Hcal_T\to[0,\infty)
\end{equation}
be a measurable loss. The fixed-prior Bayesian CVaR value is
\begin{equation}
    \CVaR_\alpha^{\mu,\Alg}(L)
    \defeq
    \inf_{s\in\R}
    \left\{
        s+\frac{1}{1-\alpha}
        \E_{\mathbb P_{\mu,\Alg}}[(L-s)_+]
    \right\}.
\end{equation}
For the fixed prior $\mu$, the optimal Bayesian CVaR is
\begin{equation}
    \Cval_\alpha(\mu)
    \defeq
    \inf_{\Alg\in\Dcal}
    \CVaR_\alpha^{\mu,\Alg}(L),
    \label{eq:fixed-prior-value}
\end{equation}
where the expectation in \eqref{eq:isdm-cvar} is taken with respect to \eqref{eq:joint-law}.

\begin{remark}[Fixed prior]
In \eqref{eq:fixed-prior-value}, $\mu$ is specified before the procedure is optimized. A lower bound for a particular prior is therefore a lower bound for a particular Bayesian decision problem. If a parameterized family of priors is optimized afterward, the resulting statement is that at least one prior in that family is hard.
\end{remark}

%% file: sections/03_generalized_fano.tex
\section{Generalized Interactive Fano Inequality}
\label{sec:generalized-fano}
This section develops the main information-theoretic inequality of the
paper. We associate a bounded function of the loss with a binary random
variable and use data processing for $f$-divergences to compare its law
under the model-dependent experiment with its law under a reference
experiment. The resulting inequality gives both lower and upper bounds
on the expected transformed loss through Bernoulli $f$-divergence
inverses.

Fix a prior $\mu$, a procedure $\Alg$, and a probability distribution $Q$ on the transcript space. Let $\phi:\R_+\to[0,1]$ be measurable and define
\begin{align}
    a_\phi
    &\defeq
    \E_{M\sim\mu}
    \E_{H_T\sim\mathbb P^{M,\Alg}}
    [\phi(L(M,H_T))],
    \\
    b_{\phi,Q}
    &\defeq
    \E_{M\sim\mu}
    \E_{H_T\sim Q}
    [\phi(L(M,H_T))].
\end{align}
For an $f$-divergence, define the average divergence
\begin{equation}
    B_f(\Alg,Q)
    \defeq
    \E_{M\sim\mu}
    D_f(\mathbb P^{M,\Alg}\|Q).
    \label{eq:section-Bf}
\end{equation}

For $b\in[0,1]$ and $B\ge0$, define the lower and upper Bernoulli
$f$-divergence inverses by
\begin{align}
    a_f^-(B;b)
    &:=
    \inf\left\{
        a\in[0,1]:
        D_f\!\left(
            \Bern(a)\middle\|\Bern(b)
        \right)
        \le B
    \right\},
    \label{eq:bernoulli-lower-inverse}
    \\
    a_f^+(B;b)
    &:=
    \sup\left\{
        a\in[0,1]:
        D_f\!\left(
            \Bern(a)\middle\|\Bern(b)
        \right)
        \le B
    \right\}.
    \label{eq:bernoulli-upper-inverse}
\end{align}

The following theorem is the main result of this section. It shows that
the expected transformed loss under the model-dependent experiment and
under the reference experiment must satisfy a Bernoulli
$f$-divergence constraint. Inverting this constraint gives a two-sided
bound on the expected transformed loss.

\begin{theorem}[Generalized interactive Fano inequality]
\label{thm:generalized-fano}
Fix a prior $\mu$, a procedure $\Alg$, a reference distribution $Q$,
a measurable function $\phi:\R_+\to[0,1]$, and an $f$-divergence.
Let
\begin{align}
    a_\phi
    &:=
    \E_{M\sim\mu}
    \E_{H_T\sim\mathbb P^{M,\Alg}}
    \left[
        \phi(L(M,H_T))
    \right],
    \\
    b_{\phi,Q}
    &:=
    \E_{M\sim\mu}
    \E_{H_T\sim Q}
    \left[
        \phi(L(M,H_T))
    \right],
\end{align}
and
\begin{equation}
    B_f(\Alg,Q)
    :=
    \E_{M\sim\mu}
    D_f\!\left(
        \mathbb P^{M,\Alg}\middle\|Q
    \right).
\end{equation}
Then
\begin{equation}
    D_f\!\left(
        \Bern(a_\phi)
        \middle\|
        \Bern(b_{\phi,Q})
    \right)
    \le
    B_f(\Alg,Q).
    \label{eq:generalized-fano}
\end{equation}
Consequently,
\begin{equation}
    a_f^-\!\left(
        B_f(\Alg,Q);b_{\phi,Q}
    \right)
    \le
    a_\phi
    \le
    a_f^+\!\left(
        B_f(\Alg,Q);b_{\phi,Q}
    \right).
    \label{eq:generalized-fano-two-sided}
\end{equation}
\end{theorem}

\begin{remark}[Choice of inversion direction]
\label{rem:inverse-direction}
Theorem~\ref{thm:generalized-fano} gives a two-sided constraint on the
expected transform $a_\phi$. The endpoint used in a particular
application depends on the direction of the desired converse.

For the CVaR lower bound developed below, we use the nondecreasing
hinge transform
\begin{equation}
    \phi_s(\ell)
    =
    \frac{(\ell-s)_+}{L_{\max}},
\end{equation}
and therefore require a lower bound on
$\E[\phi_s(L)]$. Hence only the lower endpoint
$a_f^-$ enters the CVaR converse.

The upper endpoint $a_f^+$ is retained in the general theorem because
it is useful for transforms for which an upper bound on the expected
transform yields the desired lower bound, such as
success-probability or complementary tails.
\end{remark}

\begin{remark}[Binary reduction]
The proof introduces an independent $U\sim\operatorname{Unif}[0,1]$ and sets
\begin{equation}
    V=\1\{U\le\phi(L(M,H_T))\}.
\end{equation}
Under the model-dependent joint law, $V\sim\Bern(a_\phi)$; under the reference joint law, $V\sim\Bern(b_{\phi,Q})$. Data processing gives \eqref{eq:generalized-fano}.
\end{remark}

\begin{remark}[Relation to quantile reductions]
Choosing $\phi$ as the indicator of a loss event recovers the 
single-threshold reduction used in quantile lower bounds \cite{chen2024assouad}. More
generally, Theorem~\ref{thm:generalized-fano} allows any measurable
function $\phi:\R_+\to[0,1]$. This additional flexibility permits
bounded transforms of the loss, including the hinge transform used
below to obtain the CVaR lower bound.
\end{remark}

%% file: sections/04_bayesian_cvar.tex
\section{Bayesian CVaR Lower Bounds}
\label{sec:bayesian-cvar}

This section specializes the generalized Fano inequality to fixed-prior Bayesian CVaR. For a prior $\mu$ and a procedure $\Alg$, let
\[
M\sim\mu,
\qquad
H_T\mid M\sim P^{M,\Alg}.
\]
The prior-predictive Bayesian CVaR of the loss $L(M,H_T)$ is
\begin{equation}
\CVaR_{\alpha}^{\mu,\Alg}(L)
\defeq
\inf_{s\in\mathbb{R}}
\left\{
s+\frac{1}{1-\alpha}
\E_{M\sim\mu}
\E_{H_T\sim P^{M,\Alg}}
\left[(L(M,H_T)-s)_+\right]
\right\}.
\label{eq:bayesian-cvar}
\end{equation}
Assume that $0\le L\le L_{\max}$. The infimum in \eqref{eq:bayesian-cvar} may then be restricted to $s\in[0,L_{\max}]$. For such $s$, define the normalized hinge transform
\begin{equation}
\phi_s(\ell)
\defeq
\frac{(\ell-s)_+}{L_{\max}},
\label{eq:hinge-transform}
\end{equation}
so that $\phi_s(\ell)\in[0,1]$. For a reference distribution $Q$ on the transcript space, independent of $M\sim\mu$, define
\begin{equation}
b_s(Q)
\defeq
\frac{1}{L_{\max}}
\E_{M\sim\mu}
\E_{H_T\sim Q}
\left[(L(M,H_T)-s)_+\right].
\label{eq:reference-excess-loss}
\end{equation}
Applying Theorem~\ref{thm:generalized-fano} to $\phi_s$ yields a general lower bound on Bayesian CVaR in terms of the reference excess loss $b_s(Q)$ and the prior-averaged divergence $B_f(\Alg,Q)$.

\begin{corollary}[Fixed-prior Bayesian CVaR lower bound]
\label{cor:general-cvar}
For every reference distribution $Q$ and every $\Alg\in\Dcal$,
\begin{equation}
    \CVaR_\alpha^{\mu,\Alg}(L)
    \ge
    \inf_{0\le s\le L_{\max}}
    \left\{
        s+
        \frac{L_{\max}}{1-\alpha}
        a_f^-
        \left(
            B_f(\Alg,Q);b_s(Q)
        \right)
    \right\}.
    \label{eq:general-cvar}
\end{equation}
The same bound therefore holds after taking the infimum over any
subclass of admissible procedures, including procedures satisfying
local differential privacy.
\end{corollary}

\begin{remark}
To apply \eqref{eq:general-cvar}, it is enough to lower bound $b_s(Q)$
and upper bound $B_f(\Alg,Q)$. These two quantities are then inserted
into the lower Bernoulli inverse $a_f^-$. The reference distribution
$Q$ is chosen to make both calculations tractable.
\end{remark}

\subsection{Reverse-KL Specialization}

Take $f(u)=-\log u$. Then
$D_f(P\|Q)=D_{\mathrm{KL}}(Q\|P)$. Let
$h(b)=-b\log b-(1-b)\log(1-b)$ denote the binary entropy.

To obtain explicit bounds from Corollary~\ref{cor:general-cvar}, we
next lower bound the Bernoulli inverse for reverse KL. The following
lemma shows that an upper bound on
$D_{\mathrm{KL}}(\Bern(b)\|\Bern(a))$ yields an exponential lower
bound on $a$.
\begin{lemma}[Reverse-KL Bernoulli inverse]
\label{lem:reverse-kl}
Suppose $b\in(0,1]$ and
$D_{\mathrm{KL}}(\Bern(b)\|\Bern(a))\le B$. Then
\begin{equation}
    a
    \ge
    \exp\left(
        -\frac{B+h(b)}{b}
    \right).
    \label{eq:reverse-kl}
\end{equation}
Consequently, if $b\ge\beta>0$, then
\begin{equation}
    a_f^-(B;b)
    \ge
    \frac{1}{2^{1/\beta}}
    e^{-B/\beta}.
    \label{eq:reverse-kl-beta-general}
\end{equation}
In particular, for $\beta=1/4$,
\begin{equation}
    a_f^-(B;b)
    \ge
    \frac{1}{16}e^{-4B}.
    \label{eq:reverse-kl-quarter}
\end{equation}
\end{lemma}

\begin{remark}[Dependence on the CVaR level]
The bound \eqref{eq:reverse-kl-quarter} is exponential in the
divergence budget $B$. In \eqref{eq:general-cvar}, this term is
multiplied by $(1-\alpha)^{-1}$. Balancing these two terms when
choosing the separation between the hard models yields the factor
$\log(1/(1-\alpha))$ appearing in the applications below.
\end{remark}

%% file: sections/05_method.tex
\section{Applying the Bayesian CVaR Bound}
\label{sec:methodology}

This section summarizes the steps used to apply the Bayesian CVaR bound
to a concrete statistical problem. We first fix the prior and, when
necessary, replace the original loss by a bounded lower proxy. We then
choose a reference distribution, lower bound the corresponding
reference excess loss, and upper bound the prior-averaged divergence.
These quantities are inserted into the Bernoulli inverse and then into
the CVaR bound. If the prior belongs to a parameterized family, the
separation parameter is chosen only after the fixed-prior inequality
has been established.

\subsection{Fixed prior and bounded loss}

Choose a prior $\mu$ and a loss or bounded proxy $\ell$ satisfying
\begin{equation}
    0\le \ell\le L_{\max}
    \qquad\text{and}\qquad
    \ell\le L.
    \label{eq:method-bounded}
\end{equation}
Then
\begin{equation}
    \CVaR_\alpha(L)\ge\CVaR_\alpha(\ell).
    \label{eq:method-monotone}
\end{equation}

\subsection{Reference distribution}

Choose a probability distribution $Q$ on the transcript space. For thresholds $s\in[0,L_{\max}]$, compute
\begin{equation}
    b_s(Q)
    =
    \frac{1}{L_{\max}}
    \E_{\mu,Q}[(\ell-s)_+].
\end{equation}
Choose $Q$ so that both the excess-loss expectation and the divergence bound can be evaluated explicitly.

\subsection{Divergence bound}

For the same $Q$, bound
\begin{equation}
    B_f(\Alg,Q)
    =
    \E_{M\sim\mu}
    D_f(\mathbb P^{M,\Alg}\|Q).
\end{equation}
The bound must hold for every algorithm in the class being considered.

\subsection{Bernoulli inversion and CVaR}

A lower bound
\begin{equation}
    b_s(Q)\ge\beta
\end{equation}
combined with an upper bound $B_f(\Alg,Q)\le\overline B$ gives a lower bound on $a_f^-(B_f;b_s)$ through the chosen divergence specialization. This is substituted into \eqref{eq:general-cvar}.

For reverse KL, if $b_s\ge1/4$, then
\begin{equation}
    a_f^-(B_f;b_s)
    \ge
    \frac{1}{16}e^{-4B_f}.
    \label{eq:method-reverse-kl}
\end{equation}

\begin{remark}[Fixed prior and hard-family choice]
For a fixed separation parameter, the prior is fixed and the lower bound is a statement about that specified prior. Afterward, choosing the separation parameter as a function of $\alpha$, $T$, $n$, or the privacy parameters gives an existence result: at least one prior in the chosen family satisfies the resulting optimized lower bound.
\end{remark}

%% file: sections/06_gaussian.tex
\section{Bayesian CVaR Bounds for Gaussian Models}
\label{sec:gaussian}

This section applies the general Bayesian CVaR bound to two Gaussian
statistical problems without privacy constraints. We first consider
Gaussian mean estimation and then Gaussian $K$-armed stochastic
bandits. In each case, we first derive a lower bound for every fixed
prior in a parameterized family and then choose the separation
parameter to obtain an explicit hard-prior rate. The same Gaussian
model families will be used again under LDP in
Section~\ref{sec:private-applications}.

\subsection{Gaussian mean estimation}

Let
\begin{equation}
    X_1,\ldots,X_n\stackrel{\mathrm{i.i.d.}}{\sim}\mathcal N(\theta,\sigma^2),
    \qquad
    \theta\in[-R,R],
    \label{eq:gaussian-mean-model}
\end{equation}
where $\sigma>0$. For $\Delta\in(0,R]$, fix the prior
\begin{equation}
    \mu_\Delta
    \defeq
    \frac12\delta_{-\Delta}+\frac12\delta_{+\Delta}.
    \label{eq:gaussian-mean-prior}
\end{equation}
For an estimator $\widehat\theta$, use the bounded loss
\begin{equation}
    \ell_\Delta
    \defeq
    \min\left\{
        (\widehat\theta-\theta)^2,
        \Delta^2
    \right\}.
    \label{eq:gaussian-mean-loss}
\end{equation}
Then $0\le\ell_\Delta\le\Delta^2$ and, for every $z\in\R$,
\begin{equation}
    \ell_\Delta(+\Delta,z)+\ell_\Delta(-\Delta,z)\ge\Delta^2.
    \label{eq:gaussian-mean-separation}
\end{equation}

The following theorem gives the fixed-prior Bayesian CVaR lower bound
for squared-error estimation under the two-point prior
$\mu_\Delta$. It first holds for every fixed
$\Delta\in(0,R]$; choosing $\Delta$ afterward gives the stated
dependence on $n$, $\sigma^2$, $R$, and the CVaR level $\alpha$.

\begin{theorem}[Gaussian mean estimation]
\label{thm:gaussian-mean}
For every $\Delta\in(0,R]$ and $\alpha\in(0,1)$,
\begin{equation}
\begin{aligned}
    \inf_{\widehat\theta}
    \CVaR_\alpha^{\mu_\Delta}
    \left((\widehat\theta-\theta)^2\right)
    \ge
    \Delta^2
    \min\left\{
        \frac12,
        \frac{1}{16(1-\alpha)}
        \exp\left(-\frac{4n\Delta^2}{\sigma^2}\right)
    \right\}.
\end{aligned}
\label{eq:gaussian-mean-fixed}
\end{equation}
If $\delta=1-\alpha<1/8$, choosing
\begin{equation}
    \Delta^2
    \defeq
    \min\left\{
        R^2,
        \frac{\sigma^2}{4n}
        \log\frac{1}{8\delta}
    \right\}
    \label{eq:gaussian-mean-choice}
\end{equation}
gives
\begin{equation}
    \inf_{\widehat\theta}
    \CVaR_\alpha^{\mu_\Delta}
    \left((\widehat\theta-\theta)^2\right)
    \ge
    \frac12
    \min\left\{
        R^2,
        \frac{\sigma^2}{4n}
        \log\frac{1}{8(1-\alpha)}
    \right\}.
    \label{eq:gaussian-mean-optimized}
\end{equation}
In particular,
\begin{equation}
    \inf_{\widehat\theta}
    \CVaR_\alpha^{\mu_\Delta}
    \left((\widehat\theta-\theta)^2\right)
    \gtrsim
        \min\left\{
            R^2,
            \frac{\sigma^2}{n}
            \log\frac{1}{1-\alpha}
        \right\},
\end{equation}
up to universal constants in the stated range of $\alpha$.
\end{theorem}

\begin{remark}[Relation to the mean-estimation minimax rate]
The optimized lower bound has the familiar $1/n$ squared-error dependence, with an additional logarithmic factor in the CVaR tail level. The theorem is a fixed-prior statement before the choice in \eqref{eq:gaussian-mean-choice}.
\end{remark}

\subsection{K-armed Gaussian bandits}

Let $K\ge2$. For $g\in(0,1]$ and $j\in[K]$, define the Gaussian
bandit model $M_j^g$ by
\begin{equation}
    Y_t\mid A_t=a,M_j^g
    \sim
    \begin{cases}
        \mathcal N(g,\sigma^2),& a=j,\\
        \mathcal N(0,\sigma^2),& a\ne j.
    \end{cases}
    \label{eq:gaussian-k-family}
\end{equation}
Thus arm $j$ is uniquely optimal with gap $g$. Fix the uniform prior
\begin{equation}
    \mu_g^K
    =
    \frac1K\sum_{j=1}^K\delta_{M_j^g}.
    \label{eq:gaussian-k-prior}
\end{equation}
Let $N_j(T)$ be the number of pulls of arm $j$. The pseudo-regret under
$M_j^g$ is
\begin{equation}
    R_T(M_j^g)=g\bigl(T-N_j(T)\bigr).
    \label{eq:kbandit-regret}
\end{equation}

The following theorem lower bounds the Bayesian CVaR of cumulative
pseudo-regret under the fixed prior $\mu_g^K$. The first bound holds
for every fixed gap $g\in(0,1]$. Choosing the gap afterward gives
an explicit lower bound in terms of $T$, $K$, $\sigma^2$, and
$\alpha$.

\begin{theorem}[$K$-armed Gaussian bandit]
\label{thm:gaussian-k-bandit}
For every $g\in(0,1]$ and $\alpha\in(0,1)$,
\begin{equation}
    \inf_{\Alg\in\Dcal_T}
    \CVaR_\alpha^{\mu_g^K}(R_T)
    \ge
    Tg
    \min\left\{
        \frac{K-1}{2K},
        \frac{1}{16(1-\alpha)}
        \exp\left(-\frac{2g^2T}{K\sigma^2}\right)
    \right\}.
    \label{eq:kbandit-fixed}
\end{equation}
For $1-\alpha<K/[8(K-1)]$, choosing
\begin{equation}
    g^2
    \defeq
    \min\left\{
        1,
        \frac{K\sigma^2}{2T}
        \log\frac{K}{8(K-1)(1-\alpha)}
    \right\}
    \label{eq:kbandit-gap}
\end{equation}
gives, for a suitable uniform hard prior,
\begin{equation}
    \inf_{\Alg\in\Dcal_T}
    \CVaR_\alpha^{\mu_g^K}(R_T)
    \ge
    \frac{K-1}{2K}T
    \min\left\{
        1,
        \sigma\sqrt{\frac{K}{2T}}
        \sqrt{\log\frac{K}{8(K-1)(1-\alpha)}}
    \right\},
    \label{eq:kbandit-optimized}
\end{equation}
Consequently, for fixed $K$,
\begin{equation}
    \inf_{\Alg\in\Dcal_T}
    \CVaR_\alpha^{\mu_g^K}(R_T)
    \gtrsim
        \min\left\{
            T,
            \sigma\sqrt{KT\log\frac{1}{1-\alpha}}
        \right\}
    \label{eq:kbandit-optimized-rate}
\end{equation}
up to constants depending only on $K$ and for the stated range of
$\alpha$.
\end{theorem}

\begin{remark}[Origin of the $1/K$ factor]
The uniform prior contains $K$ possible optimal arms. Under the
one-good-arm construction \eqref{eq:gaussian-k-family}, the null
reference with all arms distributed as $\mathcal N(0,\sigma^2)$ differs
from $M_j^g$ only when arm $j$ is selected. Averaging the resulting
transcript divergence over $j$ and using
$\sum_{j=1}^K N_j(T)=T$ produces the factor $1/K$.
\end{remark}

%% file: sections/07_ldp_isdm.tex
\section{Interactive Statistical Decision Making under LDP}
\label{sec:private-isdm}

This section incorporates local differential privacy into the ISDM
model. We first define the released observation and transcript laws
when each raw observation is passed through an LDP channel. We then
recall the contraction of distinguishability induced by an LDP channel
and use it together with the KL chain rule to obtain a divergence bound
for the complete adaptive released transcript. This transcript-level
bound will provide the information term in the Bayesian CVaR results
of Section~\ref{sec:private-applications}.

\subsection{Observation model under LDP}

At round $t$, given the released history
\begin{equation}
    H_{t-1}^{\mathrm{LDP}}
    =(A_1,Z_1,\ldots,A_{t-1},Z_{t-1}),
\end{equation}
the learner chooses
\begin{equation}
    A_t\sim\pi_t(\cdot\mid H_{t-1}^{\mathrm{LDP}}).
\end{equation}
The model then generates
\begin{equation}
    Y_t\sim P_t^M(\cdot\mid A_t,H_{t-1}^{\mathrm{LDP}}),
\end{equation}
and the learner observes
\begin{equation}
    Z_t\sim W_t(\cdot\mid Y_t,A_t,H_{t-1}^{\mathrm{LDP}}).
\end{equation}

We first specify the privacy requirement imposed on each released
observation.

\begin{definition}[Sequential local differential privacy]
\label{def:sequential-ldp}
The channel $W_t$ is $\varepsilon_t$-LDP if, for every history $h$,
action $a$, pair of raw observations $y,y'$, and measurable set
$C\subseteq\Zcal$,
\begin{equation}
    W_t(C\mid y,a,h)
    \le
    e^{\varepsilon_t}
    W_t(C\mid y',a,h).
    \label{eq:ldp-def}
\end{equation}
A procedure is $\epsvec$-LDP if this condition holds at every round,
where
\begin{equation}
    \epsvec=(\varepsilon_1,\ldots,\varepsilon_T).
\end{equation}
The condition is meaningful for every finite $\varepsilon_t\ge0$.
The case $\varepsilon_t=0$ forces the released distribution to be the
same for every raw observation, while larger values impose weaker
privacy constraints. In the lower bounds below, the additional
restriction $\varepsilon_t\le1$ is used only to simplify the contraction
coefficient to a constant multiple of $\varepsilon_t^2$.
\end{definition}

Integrating out the raw observation gives the released observation kernel
\begin{equation}
    \widetilde P_t^{M,\Alg}(dz\mid a,h)
    \defeq
    \int W_t(dz\mid y,a,h)
    P_t^M(dy\mid a,h).
    \label{eq:released-kernel}
\end{equation}
The released transcript is
\begin{equation}
    H_T^{\mathrm{LDP}}
    =(A_1,Z_1,\ldots,A_T,Z_T),
\end{equation}
and its distribution under $(M,\Alg)$ is denoted by
\begin{equation}
    \mathbb P_{\mathrm{LDP}}^{M,\Alg}.
\end{equation}

\begin{remark}[ISDM representation under LDP]
Equation~\eqref{eq:released-kernel} shows directly how the LDP problem
fits the ISDM model. The model still specifies an observation law, the
procedure still selects actions from the available history, and the
interaction still induces a transcript distribution. The only change is
that the learner observes $Z_t$ rather than $Y_t$.
\end{remark}

For a fixed prior $\mu$, define the Bayesian CVaR value under LDP by
\begin{equation}
    \Cval_\alpha^{\mathrm{LDP}}(\mu)
    \defeq
    \inf_{\Alg\in\Dcal_{\epsvec}^{\mathrm{LDP}}}
    \CVaR_\alpha^{\mu,\Alg}(L).
    \label{eq:ldp-value}
\end{equation}

\subsection{Divergence contraction}

The effect of an LDP channel is most naturally stated in terms of the
marginal distribution that it induces on the released observation. Let
$P$ be a distribution on the raw observation space $\Ycal$, and let
$W$ be a Markov kernel from $\Ycal$ to the released observation space
$\Zcal$. Define the induced released marginal by
\begin{equation}
    M_P^W(C)
    \defeq
    \int_{\Ycal} W(C\mid y)\,P(dy),
    \qquad C\subseteq\Zcal.
    \label{eq:induced-private-marginal}
\end{equation}
Thus, if
\begin{equation}
    Y\sim P,
    \qquad
    Z\mid Y=y\sim W(\cdot\mid y),
\end{equation}
then
\begin{equation}
    Z\sim M_P^W.
\end{equation}
Similarly, a second raw observation distribution $P'$ induces the
released marginal
\begin{equation}
    M_{P'}^W(C)
    \defeq
    \int_{\Ycal} W(C\mid y)\,P'(dy).
\end{equation}

For $\varepsilon\ge0$, define
\begin{equation}
    \kappa(\varepsilon)
    \defeq
    \min\{4,e^{2\varepsilon}\}(e^\varepsilon-1)^2.
    \label{eq:kappa-def}
\end{equation}
A contraction inequality of Duchi, Jordan, and Wainwright
\cite{duchi2018minimax} states that if $W$ is $\varepsilon$-LDP, then
for any two raw observation distributions $P$ and $P'$, the corresponding
released marginals satisfy
\begin{equation}
\begin{aligned}
    &\KL\left(M_P^W\middle\|M_{P'}^W\right)
    +
    \KL\left(M_{P'}^W\middle\|M_P^W\right)
    \\
    &\qquad\le
    \kappa(\varepsilon)\TV(P,P')^2.
\end{aligned}
\label{eq:ldp-contraction}
\end{equation}
The left-hand side is therefore a symmetrized KL divergence between
released marginals, while the total variation distance on the right-hand
side is computed from the raw observation distributions.

In the sequential setting, for a fixed history $h$ and action $a$, define
\begin{equation}
    \widetilde P_t^M(C\mid a,h)
    \defeq
    \int_{\Ycal}
        W_t(C\mid y,a,h)
        P_t^M(dy\mid a,h).
    \label{eq:conditional-induced-marginal}
\end{equation}
This is the conditional marginal distribution of the released
observation $Z_t$ under model $M$, given
$(A_t,H_{t-1})=(a,h)$. If $W_t$ is $\varepsilon_t$-LDP, then for any
two models $M$ and $M'$, \eqref{eq:ldp-contraction} gives
\begin{align}
    &\KL\left(
        \widetilde P_t^M(\cdot\mid a,h)
        \middle\|
        \widetilde P_t^{M'}(\cdot\mid a,h)
    \right)
    \nonumber\\
    &\quad+
    \KL\left(
        \widetilde P_t^{M'}(\cdot\mid a,h)
        \middle\|
        \widetilde P_t^M(\cdot\mid a,h)
    \right)
    \nonumber\\
    &\qquad\le
    \kappa(\varepsilon_t)
    \TV\left(
        P_t^M(\cdot\mid a,h),
        P_t^{M'}(\cdot\mid a,h)
    \right)^2.
    \label{eq:conditional-ldp-contraction}
\end{align}

The one-round contraction inequality controls the divergence between
two released observation laws for a fixed history and action. The next
proposition extends this control to the full adaptive transcript by
summing the conditional divergences over time.

\begin{proposition}[Adaptive transcript KL under LDP]
\label{prop:ldp-transcript}
Fix two models $M$ and $M'$ and an $\epsvec$-LDP procedure. Let
\begin{equation}
    d_t(h,a)
    \defeq
    \TV\left(
        P_t^M(\cdot\mid a,h),
        P_t^{M'}(\cdot\mid a,h)
    \right).
\end{equation}
Then
\begin{equation}
\begin{aligned}
    &\KL\left(
        \mathbb P_{\mathrm{LDP}}^{M,\Alg}
        \middle\|
        \mathbb P_{\mathrm{LDP}}^{M',\Alg}
    \right)\\
    &\qquad\le
    \E_{\mathbb P_{\mathrm{LDP}}^{M,\Alg}}
    \left[
        \sum_{t=1}^T
        \kappa(\varepsilon_t)
        d_t(H_{t-1},A_t)^2
    \right].
\end{aligned}
\label{eq:adaptive-ldp-kl}
\end{equation}
\end{proposition}

\begin{remark}[Role of the contraction coefficient]
The function $\kappa(\varepsilon)$ is the one-round contraction
coefficient in \eqref{eq:ldp-contraction}. Thus, in the applications,
the quantity
\begin{equation}
    \sum_t\kappa(\varepsilon_t)
\end{equation}
records the cumulative effect of the privacy channels. When
$\varepsilon_t\le1$, it is at most
$16\sum_t\varepsilon_t^2$.
\end{remark}

%% file: sections/08_ldp_applications.tex
\section{Bayesian CVaR Bounds under LDP}
\label{sec:private-applications}

This section combines the generalized Fano-CVaR bound with the
released transcript divergence control developed in
Section~\ref{sec:private-isdm}. The prior, loss, and model separation
used in the Gaussian problems remain unchanged; the main
difference is that the divergence between transcript laws is now
controlled through the LDP channels. We first treat Gaussian mean
estimation with heterogeneous privacy parameters, then Gaussian
$K$-armed bandits under sequential LDP, and finally give a bounded
mean-estimation result for comparison with existing heterogeneous-LDP
high-probability bounds.

\subsection{Gaussian mean estimation under heterogeneous LDP}

Consider
\begin{equation}
    X_i\stackrel{\mathrm{i.i.d.}}{\sim}\mathcal N(\theta,\sigma^2),
    \qquad
    \theta\in[-R,R],
\end{equation}
with the prior $\mu_\Delta$ from \eqref{eq:gaussian-mean-prior}. User $i$ applies an $\varepsilon_i$-LDP channel to $X_i$. Define
\begin{equation}
    \Lambda_n
    \defeq
    \sum_{i=1}^n\kappa(\varepsilon_i),
    \qquad
    \kappa(\varepsilon)
    \defeq
    \min\{4,e^{2\varepsilon}\}(e^\varepsilon-1)^2.
    \label{eq:gaussian-ldp-lambda}
\end{equation}

The following theorem gives the fixed-prior Bayesian CVaR lower bound
when the $i$th observation is released through an
$\varepsilon_i$-LDP channel. The privacy parameters enter through the
cumulative contraction quantity $\Lambda_n$. For
$\varepsilon_i\le1$, the resulting bound depends on the privacy levels
through $\sum_{i=1}^n\varepsilon_i^2$.

\begin{theorem}[Gaussian mean estimation under heterogeneous LDP]
\label{thm:gaussian-ldp-mean}
For every $\Delta\in(0,R]$ and every $\alpha\in(0,1)$,
\begin{equation}
\begin{aligned}
    &\inf_{\substack{\widehat\theta,\,W_1,\ldots,W_n\\
        W_i\text{ is }\varepsilon_i\text{-LDP for all }i}}
    \CVaR_\alpha^{\mu_\Delta}
    \left((\widehat\theta-\theta)^2\right)\\
    &\qquad\ge
    \Delta^2
    \min\left\{
        \frac12,
        \frac{1}{16(1-\alpha)}
        \exp\left(-\frac{2\Delta^2\Lambda_n}{\sigma^2}\right)
    \right\}.
\end{aligned}
\label{eq:gaussian-ldp-mean-fixed}
\end{equation}
If $\varepsilon_i\le1$ for all $i$, then
\begin{equation}
\begin{aligned}
    &\inf_{\substack{\widehat\theta,\,W_1,\ldots,W_n\\
        W_i\text{ is }\varepsilon_i\text{-LDP for all }i}}
    \CVaR_\alpha^{\mu_\Delta}
    \left((\widehat\theta-\theta)^2\right)\\
    &\qquad\ge
    \Delta^2
    \min\left\{
        \frac12,
        \frac{1}{16(1-\alpha)}
        \exp\left(-\frac{32\Delta^2}{\sigma^2}
        \sum_{i=1}^n\varepsilon_i^2\right)
    \right\}.
\end{aligned}
\label{eq:gaussian-ldp-mean-fixed-simple}
\end{equation}
Assume \(\Lambda_n>0\). Fix $\delta:=1-\alpha<1/8$. Choosing
\begin{equation}
    \Delta^2
    \defeq
    \min\left\{
        R^2,
        \frac{\sigma^2}{2\Lambda_n}
        \log\frac{1}{8\delta}
    \right\}
\end{equation}
gives, for a suitable fixed prior in the family,
\begin{equation}
\begin{aligned}
    &\inf_{\Alg\in\Dcal_{\epsvec}^{\mathrm{LDP}}}
    \CVaR_\alpha^{\mu_\Delta}
    \left((\widehat\theta-\theta)^2\right)\\
    &\qquad\ge
    \frac12
    \min\left\{
        R^2,
        \frac{\sigma^2}{2\Lambda_n}
        \log\frac{1}{8(1-\alpha)}
    \right\}.
\end{aligned}
\label{eq:gaussian-ldp-mean-optimized}
\end{equation}
In particular, when $\varepsilon_i\le1$,
\begin{equation}
    \inf_{\Alg\in\Dcal_{\epsvec}^{\mathrm{LDP}}}
    \CVaR_\alpha^{\mu_\Delta}
    \left((\widehat\theta-\theta)^2\right)
    \gtrsim
        \min\left\{
            R^2,
            \frac{\sigma^2\log(1/(1-\alpha))}
            {\sum_{i=1}^n\varepsilon_i^2}
        \right\}
    \label{eq:gaussian-ldp-mean-rate}
\end{equation}
If $\Lambda_n=0$, then $\epsilon_i=0$ for every $i$. In this case,
the fixed-prior bound~(78), with $\Delta=R$, gives
\[
\inf_{\mathrm{ALG}\in\mathcal D_{\epsilon}^{\mathrm{LDP}}}
\operatorname{CVaR}_{\alpha}^{\mu_R}
\bigl((\widehat\theta-\theta)^2\bigr)
\ge
R^2\min\left\{\frac12,\frac{1}{16(1-\alpha)}\right\}.
\]
\end{theorem}

\begin{remark}[Effect of the LDP constraint]
Theorem~\ref{thm:gaussian-mean} uses the same two-point Gaussian family and loss. The only change in the LDP proof is the bound on the transcript divergence. Without an LDP constraint the relevant reverse-KL divergence is of order $n\Delta^2/\sigma^2$, whereas under heterogeneous LDP it is controlled by $\Delta^2\Lambda_n/\sigma^2$. For small privacy parameters, $\Lambda_n\le16\sum_i\varepsilon_i^2$.
\end{remark}

\begin{remark}[Unbounded Gaussian observations]
The Gaussian observations are not assumed to be bounded. The LDP contraction in \eqref{eq:ldp-contraction} applies to arbitrary input distributions through their total variation distance. The loss is bounded by the truncation in \eqref{eq:gaussian-mean-loss}; no bounded-support assumption on the Gaussian observation is used. The mean-estimation mechanism is noninteractive across users, as in the heterogeneous-LDP formulation considered for comparison below.
\end{remark}

\begin{remark}[Comparison with bounded mean estimation]
The bounded mean-estimation result at the end of this section uses a different observation model. We include it because it matches the boundedness and heterogeneous-LDP assumptions used by Aliakbarpour et al.~\cite{aliakbarpour2025highprobability}, allowing a comparison of the dependence on $1-\alpha$ and the privacy parameters. 
\end{remark}

\subsection{Gaussian K-armed bandits under LDP}

Consider the Gaussian bandit family $\{M_j^g:j\in[K]\}$ from Section~\ref{sec:gaussian}, with prior $\mu_g^K$. The raw reward distribution of the selected arm is passed through an $\varepsilon_t$-LDP channel before it is revealed to the learner. Define
\begin{equation}
    \Lambda_T
    \defeq
    \sum_{t=1}^T\kappa(\varepsilon_t),
    \qquad
    \kappa(\varepsilon)
    \defeq
    \min\{4,e^{2\varepsilon}\}(e^\varepsilon-1)^2.
    \label{eq:gaussian-bandit-ldp-lambda}
\end{equation}

We next apply the same argument to the Gaussian bandit family from
Section~\ref{sec:gaussian}. The following theorem lower bounds the
Bayesian CVaR of cumulative pseudo-regret when the reward observed at
round $t$ is released through an $\varepsilon_t$-LDP channel. The
privacy dependence enters through the cumulative contraction quantity
$\Lambda_T$.

\begin{theorem}[Gaussian $K$-armed bandit under LDP]
\label{thm:gaussian-ldp-bandit}
For every $g\in(0,1]$ and every $\alpha\in(0,1)$,
\begin{equation}
    \inf_{\Alg\in\Dcal_{T,\epsvec}^{\mathrm{LDP}}}
    \CVaR_\alpha^{\mu_g^K}(R_T)
    \ge
    Tg
    \min\left\{
        \frac{K-1}{2K},
        \frac{1}{16(1-\alpha)}
        \exp\left(-\frac{g^2\Lambda_T}{K\sigma^2}\right)
    \right\}.
    \label{eq:gaussian-ldp-bandit-fixed}
\end{equation}
If $\varepsilon_t\le1$ for every $t$, then
\begin{equation}
    \inf_{\Alg\in\Dcal_{T,\epsvec}^{\mathrm{LDP}}}
    \CVaR_\alpha^{\mu_g^K}(R_T)
    \ge
    Tg
    \min\left\{
        \frac{K-1}{2K},
        \frac{1}{16(1-\alpha)}
        \exp\left(-\frac{16g^2}{K\sigma^2}
        \sum_{t=1}^T\varepsilon_t^2\right)
    \right\}.
    \label{eq:gaussian-ldp-bandit-simple}
\end{equation}
Assume $\Lambda_T>0$. For
$\delta:=1-\alpha<K/[8(K-1)]$, choose
\begin{equation}
    g^2
    \defeq
    \min\left\{
        1,
        \frac{K\sigma^2}{\Lambda_T}
        \log\frac{K}{8(K-1)\delta}
    \right\}.
    \label{eq:gaussian-ldp-bandit-gap}
\end{equation}
Then, for the corresponding fixed prior,
\begin{equation}
\begin{aligned}
    &\inf_{\Alg\in\Dcal_{T,\epsvec}^{\mathrm{LDP}}}
    \CVaR_\alpha^{\mu_g^K}(R_T)\\
    &\qquad\ge
    \frac{K-1}{2K}T
    \min\left\{
        1,
        \frac{\sigma\sqrt K}{\sqrt{\Lambda_T}}
        \sqrt{\log\frac{K}{8(K-1)(1-\alpha)}}
    \right\}.
\end{aligned}
\label{eq:gaussian-ldp-bandit-optimized}
\end{equation}
For $K=2$ and homogeneous $0<\epsilon_t=\epsilon\le1$,
\begin{equation}
    \inf_{\Alg\in\Dcal_{T,\epsvec}^{\mathrm{LDP}}}
    \CVaR_\alpha^{\mu_g^2}(R_T)
    \gtrsim
        \min\left\{
            T,
            \frac{\sigma}{\varepsilon}
            \sqrt{T\log\frac{1}{1-\alpha}}
        \right\}
    \label{eq:gaussian-ldp-two-rate}
\end{equation}
More generally, for homogeneous $0<\epsilon\le1$ and fixed $K$,
\begin{equation}
    \inf_{\Alg\in\Dcal_{T,\epsvec}^{\mathrm{LDP}}}
    \CVaR_\alpha^{\mu_g^K}(R_T)
    \gtrsim
        \min\left\{
            T,
            \frac{\sigma}{\varepsilon}
            \sqrt{KT\log\frac{1}{1-\alpha}}
        \right\}
    \label{eq:gaussian-ldp-k-rate}
\end{equation}
If $\Lambda_T=0$, then $\epsilon_t=0$ for every $t$. In this case,
the fixed-prior bound~(84), with $g=1$, gives
\[
\inf_{\mathrm{ALG}\in\mathcal D_{T,\epsilon}^{\mathrm{LDP}}}
\operatorname{CVaR}_{\alpha}^{\mu_1^K}(R_T)
\ge
T\min\left\{\frac{K-1}{2K},
\frac{1}{16(1-\alpha)}\right\}.
\]
\end{theorem}

\begin{remark}[Effect of the LDP constraint on the bandit bound]
Theorem~\ref{thm:gaussian-k-bandit} and Theorem~\ref{thm:gaussian-ldp-bandit} use the same Gaussian environments, prior, and pseudo-regret. Under LDP, the transcript divergence is bounded after the reward distributions are passed through the LDP channels. This replaces the usual dependence on $T$ in the information term by the cumulative privacy quantity $\Lambda_T$. For homogeneous small $\varepsilon$, this gives the factor $1/\varepsilon$ in the resulting lower bound. The rates involving $\sum_i\epsilon_t^2$ or $1/\epsilon$ apply only in the small-privacy regime and as $\epsilon\to\infty$, the LDP contraction bound becomes vacuous, and the results of Section~VI apply.
\end{remark}

\begin{remark}[Adaptive action selection]
The factor $1/K$ in the divergence budget comes from averaging over the $K$ possible optimal arms. At each round the algorithm selects one arm, so only one model in the hard family differs from the null reference on the observed reward distribution. The calculation uses the conditional LDP contraction followed by the adaptive KL chain rule.
\end{remark}

\begin{remark}[Comparison with LDP expected-regret bounds]
The factor $\sigma\sqrt{KT}/\varepsilon$ without the CVaR term is consistent with the usual statistical cost of LDP in stochastic bandits. The new feature of the present CVaR statement is the explicit factor $\sqrt{\log(1/(1-\alpha))}$.
\end{remark}

\subsection{Bounded mean estimation for comparison with heterogeneous LDP}
\label{subsec:bounded-comparison}

The next theorem uses a bounded observation model in order to compare directly with the heterogeneous-LDP mean-estimation problem studied by Aliakbarpour et al.~\cite{aliakbarpour2025highprobability}. 

Let $X_i\in[-1,1]$ and consider the two distributions
\begin{align}
    P_{+,\Delta}(X=1)&=\frac{1+\Delta}{2},
    &P_{+,\Delta}(X=-1)&=\frac{1-\Delta}{2},\\
    P_{-,\Delta}(X=1)&=\frac{1-\Delta}{2},
    &P_{-,\Delta}(X=-1)&=\frac{1+\Delta}{2},
\end{align}
whose means are $\pm\Delta$. Let
\begin{equation}
    \mu_\Delta^{\mathrm{bd}}
    =
    \frac12\delta_{P_{+,\Delta}}+
    \frac12\delta_{P_{-,\Delta}}.
\end{equation}

The following theorem applies the same Bayesian CVaR argument to this
bounded two-point model. Its purpose is to make the dependence on
$\alpha$ and the heterogeneous privacy parameters directly comparable
with existing high-probability results for the same observation model.

\begin{theorem}[Bounded mean estimation under heterogeneous LDP]
\label{thm:bounded-ldp-mean}
Let $0<\varepsilon_i\le1$. For every $\Delta\in(0,1]$,
\begin{equation}
\begin{aligned}
    &\inf_{\substack{\widehat\theta,\,W_1,\ldots,W_n\\
        W_i\text{ is }\varepsilon_i\text{-LDP for all }i}}
    \CVaR_\alpha^{\mu_\Delta^{\mathrm{bd}}}
    \left((\widehat\theta-\theta)^2\right)\\
    &\qquad\ge
    \Delta^2
    \min\left\{
        \frac12,
        \frac{1}{16(1-\alpha)}
        \exp\left(-32\Delta^2
        \sum_{i=1}^n\varepsilon_i^2\right)
    \right\}.
\end{aligned}
\label{eq:bounded-ldp-fixed}
\end{equation}
For $1-\alpha<1/8$, there is a prior in this two-point family for which
\begin{equation}
    \inf_{\Alg\in\Dcal_{\epsvec}^{\mathrm{LDP}}}
    \CVaR_\alpha^{\mu_\Delta^{\mathrm{bd}}}
    \left((\widehat\theta-\theta)^2\right)
    \ge
    \frac12
    \min\left\{
        1,
        \frac{\log\frac{1}{8(1-\alpha)}}
        {32\sum_{i=1}^n\varepsilon_i^2}
    \right\}.
    \label{eq:bounded-ldp-rate}
\end{equation}
\end{theorem}

\begin{remark}[Comparison with Aliakbarpour et al.]
Aliakbarpour et al.~\cite{aliakbarpour2025highprobability} study bounded mean estimation under heterogeneous LDP and obtain matching high-probability minimax upper and lower bounds of order
\begin{equation}
    \min\left\{
        1,
        \frac{\log(1/\beta)}{\sum_i\varepsilon_i^2}
    \right\}
\end{equation}
for the squared-error criterion when $\varepsilon_i\le1$. Setting $\beta=1-\alpha$ gives the same dependence as \eqref{eq:bounded-ldp-rate}, up to universal constants. 
\end{remark}

\begin{remark}[Role of the bounded comparison]
The main applications of this paper use Gaussian observations. The bounded mean-estimation proposition is included separately because it permits a direct comparison with the heterogeneous-LDP minimax quantile result of \cite{aliakbarpour2025highprobability}.
\end{remark}

\subsection{Summary of the explicit bounds}

Table~\ref{tab:explicit-bounds} summarizes the bounds derived and the bounded mean-estimation comparison. In the table, $\delta=1-\alpha$, $R$ is the bound on the Gaussian mean parameter in the estimation problems, and $\sigma^2$ is the common observation variance. For the LDP rows, $\varepsilon_i$ or $\varepsilon_t$ denotes the corresponding privacy parameter.

\begin{table}[H]
\centering
\caption{Explicit Bayesian CVaR lower bounds derived in this paper. Each bound is first stated for a fixed prior; the displayed rate is obtained by choosing the separation parameter afterward. Universal numerical constants are suppressed in the rate column.}
\label{tab:explicit-bounds}
\renewcommand{\arraystretch}{1.2}
\begin{tabularx}{\columnwidth}{>{\raggedright\arraybackslash}p{2.25cm}>{\raggedright\arraybackslash}p{2.05cm}X}
\toprule
\textbf{Problem} & \textbf{Setting} & \textbf{CVaR lower-bound dependence} \\
\midrule
Gaussian mean estimation & $X_i\sim\mathcal N(\theta,\sigma^2)$, two-point prior & $\displaystyle \Omega\!\left(\min\left\{R^2,\frac{\sigma^2\log(1/\delta)}{n}\right\}\right)$ \\
$K$-armed Gaussian bandit & Gaussian rewards, uniform $K$-point prior & $\displaystyle \Omega\!\left(\min\left\{T,\sigma\sqrt{KT\log(1/\delta)}\right\}\right)$ \\
Gaussian mean estimation under heterogeneous LDP & $\varepsilon_i$-LDP, two-point prior & $\displaystyle \Omega\!\left(\min\left\{R^2,\frac{\sigma^2\log(1/\delta)}{\sum_i\varepsilon_i^2}\right\}\right)$ \\
Gaussian $K$-armed bandit under LDP & $\varepsilon_t$-LDP, uniform $K$-point prior & $\displaystyle \Omega\!\left(\min\left\{T,\frac{\sigma}{\varepsilon}\sqrt{KT\log(1/\delta)}\right\}\right)$ for $\varepsilon_t=\varepsilon$ \\
Bounded mean estimation under heterogeneous LDP & $X_i\in[-1,1]$, two-point prior & $\displaystyle \Omega\!\left(\min\left\{1,\frac{\log(1/\delta)}{\sum_i\varepsilon_i^2}\right\}\right)$ \\
\bottomrule
\end{tabularx}
\end{table}

\begin{remark}[Gaussian and bounded mean estimation under LDP]
The main mean-estimation theorem uses Gaussian observations, while the bounded mean-estimation theorem is included for direct comparison with Aliakbarpour et al.~\cite{aliakbarpour2025highprobability}. The latter uses the same bounded observation model and heterogeneous privacy parameters as that work. The Gaussian theorem shows that the LDP contraction step does not require the observations themselves to have bounded support; the squared-error loss is bounded in the proof by truncation.
\end{remark}

%% file: sections/09_related_work.tex
\section{Relation to Prior Work}
\label{sec:related-work}
This section places the results of the paper relative to the main
information-theoretic lower-bound frameworks that are closest to our
setting. We distinguish the present fixed-prior Bayesian CVaR criterion
from interactive minimax lower bounds, Bayesian expected-risk bounds,
minimax quantile bounds, and lower bounds under local differential
privacy. We then discuss the literature on private estimation and
bandits and clarify how CVaR is used in the bandit setting.

Table \ref{tab:related-work}
compares the present results with the closest information-theoretic lower-bound frameworks. The two preliminary conference versions of parts of this article are not treated as external prior work; their relationship to this journal article is stated in the Introduction.

\begin{table}[t]
\centering
\caption{Comparison with related lower-bound results. Fixed-prior Bayesian CVaR refers to a criterion of the form $\inf_{\Alg}\CVaR_\alpha^{\mu,\Alg}(L)$ for a specified prior $\mu$.}
\label{tab:related-work}
\renewcommand{\arraystretch}{1.15}
\begin{tabularx}{\columnwidth}{>{\raggedright\arraybackslash}p{2.35cm}>{\raggedright\arraybackslash}p{2.2cm}X>{\raggedright\arraybackslash}p{1.9cm}}
\toprule
\textbf{Work} & \textbf{Setting} & \textbf{Lower-bound criterion} & \textbf{Privacy} \\
\midrule
Foster et al.~\cite{foster2021statistical,foster2023tight} & Interactive decision making & DEC-based minimax risk and regret characterizations & None \\
Chen et al.~\cite{chen2024assouad} & Interactive decision making & Interactive Fano; quantile PAC-DEC and quantile regret-DEC & None \\
Esposito et al.~\cite{esposito2024bayesianrisk} & Bayesian estimation & Bayesian expected-risk lower bounds for general information measures & Via contraction \\
Ma et al.~\cite{ma2024highprobability} & Statistical estimation & Explicit minimax quantile lower bounds & None \\
Chen--Rakhlin~\cite{chen2025changing} & Interactive decision making & DEC lower bounds, including quantile-based quantities & LDP \\
Basu et al.; Ren et al.~\cite{basu2020privacy,ren2020bandits} & Stochastic bandits & Expected-regret lower bounds & LDP and DP \\
Aliakbarpour et al.~\cite{aliakbarpour2025highprobability} & Mean estimation & Explicit high-probability minimax rates & Heterogeneous LDP \\
Bongole et al.~\cite{bongole2026quantile} & ISDM & Explicit minimax quantile lower bounds & Mutual-information privacy \\
Tan--Weng~\cite{tan2023cvar} & Bandits & CVaR of regret; algorithmic upper bounds & None \\
\textbf{Present article} & ISDM & \textbf{Fixed-prior Bayesian CVaR lower bounds} & \textbf{LDP} \\
\bottomrule
\end{tabularx}
\end{table}

\subsection{Interactive Fano and DEC}

Foster et al.~\cite{foster2021statistical,foster2023tight} develop the Decision-Estimation Coefficient (DEC) as a general complexity measure for interactive decision making. Chen et al.~\cite{chen2024assouad} develop interactive versions of Fano, Le Cam, and Assouad methods and connect them to DEC-based lower bounds. Their framework also gives quantile versions of DEC, providing high-probability lower bounds for interactive decision problems.

The generalized Fano inequality used here addresses a different question. For a specified prior and bounded transform of the loss, it gives a Bernoulli $f$-divergence inequality directly for the prior-predictive experiment. The hinge transform gives a fixed-prior Bayesian CVaR lower bound. Thus the distinction is between a broad characterization of interactive minimax complexity and a direct lower bound for a specified Bayesian risk functional.

\subsection{Bayesian information-theoretic lower bounds}

Esposito, Vandenbroucque, and Gastpar \cite{esposito2024bayesianrisk} develop Bayesian risk lower bounds using a broad class of information measures, including $\varphi$-divergences, R\'enyi information, and Sibson mutual information. They use data processing and variational arguments to obtain estimator-independent lower bounds and discuss the effect of privatization through strong data-processing inequalities.

The present work also uses an information measure to lower bound Bayesian risk, but the setting and risk criterion are different. We allow the transcript distribution to depend on an adaptive decision rule and apply a randomized binary reduction to a bounded transform of the loss. The main risk functional is CVaR, and its level $\alpha$ remains explicit in the final bound.

\subsection{Minimax quantile lower bounds}

Ma, Verchand, and Samworth \cite{ma2024highprobability} develop high-probability variants of Le Cam and Fano lower bounds and introduce minimax quantiles. Chen et al.~\cite{chen2024assouad} develop quantile DEC within the interactive Fano framework. Bongole et al.~\cite{bongole2026quantile} develop explicit minimax quantile lower bounds for ISDM and study mutual-information privacy.

For a fixed random variable $X$,
\begin{equation}
    \CVaR_\alpha(X)\ge\VaR_\alpha(X).
\end{equation}
Therefore, if the same probability law and model are being considered, a quantile lower bound can give a CVaR lower bound. The minimax and fixed-prior Bayesian criteria are nevertheless different because
\begin{equation}
    \inf_{\Alg}\sup_{M\in\Mcal}\VaR_\alpha^{M,\Alg}(L)
\end{equation}
uses the worst model, whereas
\begin{equation}
    \inf_{\Alg}\CVaR_\alpha^{\mu,\Alg}(L)
\end{equation}
uses the specified prior to form the loss distribution before CVaR is evaluated.

\subsection{Decision making under LDP}

Duchi, Jordan, and Wainwright \cite{duchi2018minimax} establish the LDP contraction inequality used in the present work. Chen and Rakhlin \cite{chen2025changing} study interactive decision making under LDP using variants of the DEC and include quantile-based private DEC quantities. Thus the compatibility of adaptive decision making with LDP and the corresponding information contractions are already established.

Our use of these results is to control the divergence term in the generalized Bayesian CVaR inequality. This gives explicit lower bounds in the problem parameters and in $\alpha$, rather than a bound left in terms of a DEC quantity.

\subsection{LDP lower bounds for estimation and bandits}

For mean estimation, Aliakbarpour et al.~\cite{aliakbarpour2025highprobability} study bounded observations under heterogeneous LDP and derive matching high-probability minimax bounds. Joseph et al.~\cite{joseph2019locally} study Gaussian mean estimation under LDP, where the unbounded data domain creates additional difficulties. The Gaussian LDP mean-estimation theorem in this paper uses a two-point argument and the LDP contraction inequality; it is stated without claiming that the resulting bound is a complete characterization of the unbounded Gaussian problem.

For stochastic bandits, Basu, Dimitrakakis, and Tossou \cite{basu2020privacy}, Ren et al.~\cite{ren2020bandits}, and Tao et al.~\cite{tao2022optimal} establish regret lower and upper bounds under privacy constraints, including LDP. These results are primarily expectation-based. The present bandit theorem applies CVaR directly to the cumulative pseudo-regret under a fixed prior and displays the additional dependence on the CVaR level.

\subsection{CVaR in bandit problems}

CVaR has been studied as a risk-sensitive objective in bandits and reinforcement learning. Baudry et al.~\cite{baudry2021thompson} study bandit policies based on risk measures of individual arm distributions. Tan and Weng \cite{tan2023cvar} study the CVaR of the regret random variable itself and derive algorithm-specific CVaR-regret upper bounds.

The present result concerns the information-theoretic lower bound for the latter type of random variable, with an additional Bayesian prior on the bandit environment. In particular, the CVaR is taken after the model has been drawn from the fixed prior and the adaptive transcript has been generated.

%% file: sections/10_discussion.tex
\section{Discussion}
\label{sec:discussion}
This section discusses the interpretation and scope of the preceding
results. We first clarify the role of the fixed prior and the choice of
$f$-divergence, then explain how LDP changes the information term in
the lower bound, and finally discuss the relation between the Gaussian
and bounded examples, the connection with DEC-based results, and the
main limitations of the present analysis.

\subsection{Fixed-prior Bayesian CVaR}

The basic quantity in this paper is
\begin{equation}
    \inf_{\Alg\in\Dcal}
    \CVaR_\alpha^{\mu,\Alg}(L)
\end{equation}
for a specified prior $\mu$. The lower bounds in Sections~\ref{sec:gaussian} and~\ref{sec:private-applications} are first stated for fixed priors such as $\mu_\Delta$ or $\mu_g^K$. Only after these statements are proved do we choose a separation parameter to obtain a rate in $n$, $T$, or $\alpha$.

This distinction is useful when comparing with minimax quantiles. A minimax quantile is unaffected by how much prior probability is assigned to a difficult model. A Bayesian CVaR need not be. For example, if a difficult model has prior mass $\rho$, its contribution to the upper $1-\alpha$ tail can be small when $\rho\ll1-\alpha$. The prior therefore remains part of the statistical problem throughout the Bayesian calculation.

\subsection{Choice of divergence}

The generalized Fano inequality permits different $f$-divergences. The choice affects the Bernoulli inverse and therefore the dependence on $\alpha$. Hellinger distance gives a convenient explicit inverse. Reverse KL gives an exponential inverse as in \eqref{eq:reverse-kl-quarter}, which leads to $\log(1/(1-\alpha))$ after the hard-instance separation is chosen.

The divergence choice does not automatically improve the statistical rate of a problem. Its role is to determine whether the lower-bound calculation retains the desired dependence on the tail probability.

\subsection{Effect of local differential privacy}

Under LDP, the model, prior, and loss need not be changed. The raw observation is passed through an LDP channel and the learner acts on the resulting observation. The LDP contraction changes the bound on the transcript divergence.

For Gaussian mean estimation, the reverse-KL budget without privacy constraints is of order
\begin{equation}
    \frac{n\Delta^2}{\sigma^2},
\end{equation}
whereas the heterogeneous-LDP calculation gives a budget of order
\begin{equation}
    \frac{\Delta^2}{\sigma^2}
    \sum_i\kappa(\varepsilon_i).
\end{equation}
For Gaussian bandits, the same replacement occurs in the adaptive transcript calculation. Thus the statistical model and loss calculation stay the same while the information bound changes.

\subsection{Gaussian and bounded mean estimation under LDP}

The main mean-estimation theorem uses Gaussian observations in order to remain consistent with the Gaussian model used elsewhere in the paper. The observation is unbounded, but the LDP contraction depends on total variation and does not require bounded support.

The bounded mean-estimation theorem is included separately because it gives a direct comparison with the heterogeneous-LDP minimax quantile result of Aliakbarpour et al.~\cite{aliakbarpour2025highprobability}. In that comparison the statistical model, boundedness assumption, and heterogeneous privacy parameters agree, so the two rates can be compared directly.

\subsection{Relation between explicit bounds and DEC}

DEC-based results and the bounds developed here serve different purposes. DEC provides a broad complexity measure for interactive learning and can be used to characterize large families of problems. The present generalized Fano inequality instead gives a direct expression for a specified Bayesian CVaR. Once the reference excess loss and divergence are evaluated for a concrete problem, the dependence on $\alpha$ and the model parameters is explicit.

The LDP bandit result is an example of this distinction. Existing LDP bandit lower bounds establish the cost of privacy for expected regret, while the private DEC framework provides a general high-probability lower-bound formulation. The theorem here evaluates the corresponding information constraint directly in a standard Gaussian bandit and obtains an explicit CVaR lower bound.

\subsection{Limitations and extensions}

The privacy results use pure LDP and the KL contraction in \eqref{eq:ldp-contraction}. Extending the analysis to approximate LDP or other privacy definitions requires the corresponding divergence inequalities. The present work gives lower bounds and does not establish matching Bayesian CVaR upper bounds.

The fixed priors used for the explicit calculations are finite. The generalized Fano and CVaR inequalities themselves allow arbitrary priors. For nonuniform or continuous priors, useful bounds will require direct control of the prior-weighted excess loss and divergence terms. This may lead to bounds that depend explicitly on prior mass near difficult subsets of the model class.

%% file: sections/11_conclusion.tex
\section{Conclusion}
\label{sec:conclusion}

We developed an information-theoretic lower-bound method for the Bayesian CVaR of losses in interactive statistical decision making. The generalized Fano inequality reduces a bounded function of the loss to a Bernoulli $f$-divergence. A bounded hinge function then gives a direct lower bound on the CVaR under a specified prior.

We used this result to obtain explicit bounds for Gaussian mean estimation and Gaussian stochastic bandits. 

We also incorporated local differential privacy at the observation stage of the ISDM model. Established LDP contraction inequalities, combined with the KL chain rule in adaptive problems, give the divergence bounds required by the same generalized Fano inequality. This yields explicit Bayesian CVaR bounds for Gaussian mean estimation under heterogeneous LDP and Gaussian stochastic bandits under LDP.

Several questions remain open. It is natural to study nonuniform and continuous priors, where the lower bound may depend explicitly on how the prior places mass on difficult models. Extensions to other privacy definitions and to other reinforcement learning settings are also possible. A further question is to identify statistical problems for which a direct Bayesian CVaR converse yields a substantially sharper statement than a corresponding quantile lower bound.

%% file: appendices/A_generalized_fano.tex
\section{Proof of the Generalized Fano Inequality and CVaR Corollary}
\label{app:generalized-fano}

This appendix proves the two main general results used throughout the
paper. We first prove the generalized Fano inequality by reducing the
transformed loss to a Bernoulli random variable and applying
$f$-divergence data processing. We then specialize the lower side of
the resulting inversion to the hinge transform used in the Bayesian
CVaR corollary.

\begin{proof}[Proof of Theorem~\ref{thm:generalized-fano}]
Fix a prior $\mu$, a procedure $\Alg$, a reference distribution $Q$,
and a measurable function $\phi:\R_+\to[0,1]$. We prove first the
Bernoulli divergence inequality and then its two-sided inversion.

\paragraph*{1) Define the model-dependent and reference experiments}

Define probability measures $P_0$ and $P_1$ on
$\Mcal\times\Hcal_T$ by
\begin{align}
    P_0(dM,dh)
    &\defeq
    \mu(dM)\mathbb P^{M,\Alg}(dh),
    \label{eq:proof-fano-P0}
    \\
    P_1(dM,dh)
    &\defeq
    \mu(dM)Q(dh).
    \label{eq:proof-fano-P1}
\end{align}
Under $P_0$, the model is drawn from the prior and the transcript is
generated according to the model-dependent experiment
$\mathbb P^{M,\Alg}$. Under $P_1$, the model has the same prior
$\mu$, but the transcript is generated from the reference distribution
$Q$, independently of $M$.

\paragraph*{2) Introduce the randomized binary statistic}

Let
\begin{equation}
    U\sim\operatorname{Unif}[0,1]
\end{equation}
be independent of $(M,H_T)$ under both experiments. To include this
auxiliary random variable explicitly, define the product measures
\begin{equation}
    \overline P_0
    \defeq
    P_0\otimes\operatorname{Unif}[0,1],
    \qquad
    \overline P_1
    \defeq
    P_1\otimes\operatorname{Unif}[0,1].
    \label{eq:proof-fano-augmented-measures}
\end{equation}
Now define
\begin{equation}
    V
    \defeq
    \mathbf 1
    \left\{
        U\le
        \phi\bigl(L(M,H_T)\bigr)
    \right\}.
    \label{eq:proof-fano-V}
\end{equation}
Since $\phi$ takes values in $[0,1]$, this construction is well
defined.

\paragraph*{3) Identify the Bernoulli parameters}

Conditioning on $(M,H_T)$ under $\overline P_0$ gives
\begin{align}
    \overline P_0(V=1)
    &=
    \E_{\overline P_0}
    \left[
        \Pr
        \left(
            U\le\phi(L(M,H_T))
            \mid M,H_T
        \right)
    \right]
    \nonumber\\
    &=
    \E_{\overline P_0}
    \left[
        \phi(L(M,H_T))
    \right]
    \nonumber\\
    &=
    \E_{M\sim\mu}
    \E_{H_T\sim\mathbb P^{M,\Alg}}
    \left[
        \phi(L(M,H_T))
    \right]
    \nonumber\\
    &=
    a_\phi.
    \label{eq:proof-fano-a}
\end{align}
Similarly,
\begin{align}
    \overline P_1(V=1)
    &=
    \E_{M\sim\mu}
    \E_{H_T\sim Q}
    \left[
        \phi(L(M,H_T))
    \right]
    \nonumber\\
    &=
    b_{\phi,Q}.
    \label{eq:proof-fano-b}
\end{align}
Therefore,
\begin{equation}
    \Law_{\overline P_0}(V)
    =
    \Bern(a_\phi),
    \qquad
    \Law_{\overline P_1}(V)
    =
    \Bern(b_{\phi,Q}).
    \label{eq:proof-fano-bernoulli-laws}
\end{equation}

\paragraph*{4) Apply data processing}

The mapping
\begin{equation}
    (M,H_T,U)
    \longmapsto
    V
\end{equation}
is a measurable Markov kernel from
$\Mcal\times\Hcal_T\times[0,1]$ to $\{0,1\}$. Hence the
data-processing inequality for $f$-divergences gives
\begin{align}
    D_f
    \left(
        \Bern(a_\phi)
        \middle\|
        \Bern(b_{\phi,Q})
    \right)
    &\le
    D_f
    \left(
        \overline P_0
        \middle\|
        \overline P_1
    \right).
    \label{eq:proof-fano-data-processing}
\end{align}

Since the same independent uniform distribution is appended to both
$P_0$ and $P_1$, tensoring with this common factor does not change the
$f$-divergence. Thus,
\begin{equation}
    D_f
    \left(
        \overline P_0
        \middle\|
        \overline P_1
    \right)
    =
    D_f(P_0\|P_1).
    \label{eq:proof-fano-product-divergence}
\end{equation}

\paragraph*{5) Evaluate the divergence between the two joint experiments}

The model marginal is $\mu$ under both $P_0$ and $P_1$. Conditional
on $M$, the transcript distributions are
$\mathbb P^{M,\Alg}$ and $Q$, respectively. Therefore the
$f$-divergence decomposes as
\begin{align}
    D_f(P_0\|P_1)
    &=
    \E_{M\sim\mu}
    D_f
    \left(
        \mathbb P^{M,\Alg}
        \middle\|
        Q
    \right)
    \nonumber\\
    &=
    B_f(\Alg,Q).
    \label{eq:proof-fano-average-divergence}
\end{align}
Combining
\eqref{eq:proof-fano-data-processing},
\eqref{eq:proof-fano-product-divergence}, and
\eqref{eq:proof-fano-average-divergence} yields
\begin{equation}
    D_f
    \left(
        \Bern(a_\phi)
        \middle\|
        \Bern(b_{\phi,Q})
    \right)
    \le
    B_f(\Alg,Q),
\end{equation}
which proves \eqref{eq:generalized-fano}.

\paragraph*{6) Invert the Bernoulli divergence constraint}

Fix
\begin{equation}
    B
    =
    B_f(\Alg,Q),
    \qquad
    b
    =
    b_{\phi,Q}.
\end{equation}
The inequality just proved implies that
\begin{equation}
    a_\phi
    \in
    \left\{
        a\in[0,1]:
        D_f
        \left(
            \Bern(a)\middle\|\Bern(b)
        \right)
        \le B
    \right\}.
    \label{eq:proof-fano-feasible-set}
\end{equation}
By definition,
\begin{align}
    a_f^-(B;b)
    &=
    \inf
    \left\{
        a\in[0,1]:
        D_f
        \left(
            \Bern(a)\middle\|\Bern(b)
        \right)
        \le B
    \right\},
    \\
    a_f^+(B;b)
    &=
    \sup
    \left\{
        a\in[0,1]:
        D_f
        \left(
            \Bern(a)\middle\|\Bern(b)
        \right)
        \le B
    \right\}.
\end{align}
Since $a_\phi$ belongs to the feasible set in
\eqref{eq:proof-fano-feasible-set}, it follows immediately that
\begin{equation}
    a_f^-
    \left(
        B_f(\Alg,Q);b_{\phi,Q}
    \right)
    \le
    a_\phi
    \le
    a_f^+
    \left(
        B_f(\Alg,Q);b_{\phi,Q}
    \right).
\end{equation}
This proves \eqref{eq:generalized-fano-two-sided}.
\end{proof}

\begin{proof}[Proof of Corollary~\ref{cor:general-cvar}]
Fix a procedure $\Alg$ and a reference distribution $Q$. Assume
\begin{equation}
    0\le L\le L_{\max}.
\end{equation}
For $s\in[0,L_{\max}]$, apply
Theorem~\ref{thm:generalized-fano} with the hinge transform
\begin{equation}
    \phi_s(\ell)
    =
    \frac{(\ell-s)_+}{L_{\max}}.
    \label{eq:proof-cvar-hinge}
\end{equation}

\paragraph*{1) Identify the two Bernoulli parameters}

Under the model-dependent Bayesian experiment,
\begin{align}
    a_{\phi_s}
    &=
    \E_{M\sim\mu}
    \E_{H_T\sim\mathbb P^{M,\Alg}}
    \left[
        \frac{
            (L(M,H_T)-s)_+
        }{
            L_{\max}
        }
    \right]
    \nonumber\\
    &=
    \frac{1}{L_{\max}}
    \E_{\mu,\Alg}
    \left[
        (L-s)_+
    \right].
    \label{eq:proof-cvar-a}
\end{align}
Under the reference experiment,
\begin{align}
    b_{\phi_s,Q}
    &=
    \frac{1}{L_{\max}}
    \E_{M\sim\mu}
    \E_{H_T\sim Q}
    \left[
        (L(M,H_T)-s)_+
    \right]
    \nonumber\\
    &=
    b_s(Q).
    \label{eq:proof-cvar-b}
\end{align}

\paragraph*{2) Use the lower Bernoulli inverse}

The two-sided conclusion of
Theorem~\ref{thm:generalized-fano} gives, in particular,
\begin{equation}
    a_{\phi_s}
    \ge
    a_f^-
    \left(
        B_f(\Alg,Q);b_s(Q)
    \right).
\end{equation}
Substituting \eqref{eq:proof-cvar-a} yields
\begin{equation}
    \E_{\mu,\Alg}
    \left[
        (L-s)_+
    \right]
    \ge
    L_{\max}
    a_f^-
    \left(
        B_f(\Alg,Q);b_s(Q)
    \right).
    \label{eq:proof-cvar-hinge-lower}
\end{equation}

\paragraph*{3) Substitute into the CVaR representation}

By the Rockafellar--Uryasev representation,
\begin{equation}
    \CVaR_\alpha^{\mu,\Alg}(L)
    =
    \inf_{s\in\R}
    \left\{
        s+
        \frac{1}{1-\alpha}
        \E_{\mu,\Alg}
        \left[
            (L-s)_+
        \right]
    \right\}.
    \label{eq:proof-cvar-RU}
\end{equation}
Because $0\le L\le L_{\max}$, the infimum in
\eqref{eq:proof-cvar-RU} may be restricted to
$s\in[0,L_{\max}]$.

Indeed, if $s<0$, then $L-s\ge0$ almost surely, and hence
\begin{align}
    s+
    \frac{1}{1-\alpha}
    \E_{\mu,\Alg}
    \left[
        (L-s)_+
    \right]
    &=
    s+
    \frac{
        \E_{\mu,\Alg}[L]-s
    }{
        1-\alpha
    }
    \nonumber\\
    &=
    \frac{
        \E_{\mu,\Alg}[L]
        -
        \alpha s
    }{
        1-\alpha
    }
    \nonumber\\
    &\ge
    \frac{
        \E_{\mu,\Alg}[L]
    }{
        1-\alpha
    },
\end{align}
which is the value of the same objective at $s=0$.

If $s>L_{\max}$, then $(L-s)_+=0$ almost surely, so the objective is
equal to $s>L_{\max}$, whereas at $s=L_{\max}$ it is equal to
$L_{\max}$. Hence no minimizer lies outside
$[0,L_{\max}]$.

Therefore,
\begin{equation}
    \CVaR_\alpha^{\mu,\Alg}(L)
    =
    \inf_{0\le s\le L_{\max}}
    \left\{
        s+
        \frac{1}{1-\alpha}
        \E_{\mu,\Alg}
        \left[
            (L-s)_+
        \right]
    \right\}.
    \label{eq:proof-cvar-restricted}
\end{equation}
Applying \eqref{eq:proof-cvar-hinge-lower} for every
$s\in[0,L_{\max}]$ gives
\begin{equation}
    \CVaR_\alpha^{\mu,\Alg}(L)
    \ge
    \inf_{0\le s\le L_{\max}}
    \left\{
        s+
        \frac{L_{\max}}{1-\alpha}
        a_f^-
        \left(
            B_f(\Alg,Q);b_s(Q)
        \right)
    \right\}.
\end{equation}
This proves the corollary.
\end{proof}

%% file: appendices/B_divergence.tex
\section{Divergence Calculations}
\label{app:divergence}

This appendix proves the reverse-KL Bernoulli inversion used in the
Gaussian and LDP applications.

\begin{proof}[Proof of Lemma~\ref{lem:reverse-kl}]
For $f(u)=-\log u$,
\begin{equation}
    D_f(\Bern(a)\|\Bern(b))
    =
    \KL(\Bern(b)\|\Bern(a)).
\end{equation}
If the latter is at most $B$, then
\begin{align}
    B
    &\ge
    b\log\frac ba+(1-b)\log\frac{1-b}{1-a}\\
    &=-h(b)-b\log a-(1-b)\log(1-a)\\
    &\ge -h(b)-b\log a.
\end{align}
Thus
\begin{equation}
    -b\log a\le B+h(b),
\end{equation}
which gives \eqref{eq:reverse-kl}. If $b\ge\beta$, then $h(b)\le\log2$ and $1/b\le1/\beta$, so
\begin{equation}
    a\ge\exp\left(-\frac{B+\log2}{\beta}\right).
\end{equation}
For $\beta=1/4$, this is \eqref{eq:reverse-kl-quarter}. 
\end{proof}

%% file: appendices/C_gaussian.tex
\section{Proofs for the Gaussian Applications}
\label{app:gaussian}

This appendix proves the lower bounds of
Section~\ref{sec:gaussian}. We treat Gaussian mean estimation first and
then the Gaussian $K$-armed bandit. In both proofs, the fixed-prior
bound is established before the separation parameter is chosen.

\subsection{Gaussian mean estimation}

\begin{proof}[Proof of Theorem~\ref{thm:gaussian-mean}]
Fix $n\ge1$, $\sigma>0$, and $\alpha\in(0,1)$. We first fix an
arbitrary separation
\begin{equation}
    \Delta\in(0,R]
\end{equation}
and the corresponding prior
\begin{equation}
    \mu_\Delta
    =
    \frac12\delta_{-\Delta}+\frac12\delta_{+\Delta}.
\end{equation}
We also fix an arbitrary estimator $\widehat\theta=\widehat\theta(X^n)$.
The argument below will hold for every estimator.

\paragraph*{1) Define a bounded loss}
Define
\begin{equation}
    \ell_\Delta(\theta,\widehat\theta)
    :=
    \min\left\{
        (\widehat\theta-\theta)^2,\Delta^2
    \right\}.
\end{equation}
Then
\begin{equation}
    0\le \ell_\Delta\le \Delta^2.
\end{equation}
Thus $L_{\max}=\Delta^2$. Since
$\ell_\Delta\le(\widehat\theta-\theta)^2$ pointwise, monotonicity of
CVaR gives
\begin{equation}
    \CVaR_\alpha^{\mu_\Delta}
    \left((\widehat\theta-\theta)^2\right)
    \ge
    \CVaR_\alpha^{\mu_\Delta}(\ell_\Delta).
    \label{eq:gaussian-mean-proof-cvar-monotonicity}
\end{equation}

\paragraph*{2) Lower bound the reference excess loss}
For $z\in\R$, let
\begin{equation}
    \ell_+(z):=\min\{(z-\Delta)^2,\Delta^2\},
    \qquad
    \ell_-(z):=\min\{(z+\Delta)^2,\Delta^2\}.
\end{equation}
Since the two parameter values are separated by $2\Delta$,
\begin{equation}
    \max\{|z-\Delta|,|z+\Delta|\}\ge\Delta
\end{equation}
for every $z$. Hence
\begin{equation}
    \max\{\ell_+(z),\ell_-(z)\}=\Delta^2.
\end{equation}
It follows that, for every $s\in[0,\Delta^2]$,
\begin{equation}
    (\ell_+(z)-s)_++(\ell_-(z)-s)_+
    \ge \Delta^2-s.
    \label{eq:gaussian-mean-proof-separation}
\end{equation}

Take the reference distribution to be
\begin{equation}
    Q:=\mathcal N(-\Delta,\sigma^2)^{\otimes n}.
\end{equation}
Then
\begin{align}
    b_s(Q)
    &
    =\frac{1}{2\Delta^2}
    \E_Q\left[
        (\ell_+(\widehat\theta)-s)_+
        +(\ell_-(\widehat\theta)-s)_+
    \right]
    \nonumber\\
    &\ge
    \frac{\Delta^2-s}{2\Delta^2}
    =
    \frac12\left(1-\frac{s}{\Delta^2}\right).
    \label{eq:gaussian-mean-proof-bs}
\end{align}
Therefore,
\begin{equation}
    0\le s\le\frac{\Delta^2}{2}
    \quad\Longrightarrow\quad
    b_s(Q)\ge\frac14.
\end{equation}

\paragraph*{3) Compute the reverse-KL information budget}
Choose
\begin{equation}
    f(u)=-\log u,
\end{equation}
so that
\begin{equation}
    D_f(P\|Q)=\KL(Q\|P).
\end{equation}
Because $Q$ is the distribution under $-\Delta$,
\begin{align}
    B_f
    &
    =\frac12\KL(Q\|Q)
      +\frac12\KL(Q\|P_+)
    \nonumber\\
    &
    =\frac12\KL(P_-\|P_+),
\end{align}
where
\begin{equation}
    P_-=\mathcal N(-\Delta,\sigma^2)^{\otimes n},
    \qquad
    P_+=\mathcal N(+\Delta,\sigma^2)^{\otimes n}.
\end{equation}
For equal-variance Gaussians,
\begin{equation}
    \KL\left(\mathcal N(m_0,\sigma^2)\middle\|\mathcal N(m_1,\sigma^2)\right)
    =\frac{(m_0-m_1)^2}{2\sigma^2}.
\end{equation}
Thus
\begin{equation}
    \KL\left(\mathcal N(-\Delta,\sigma^2)\middle\|\mathcal N(+\Delta,\sigma^2)\right)
    =\frac{2\Delta^2}{\sigma^2}.
\end{equation}
By additivity over the independent observations,
\begin{equation}
    \KL(P_-\|P_+)=\frac{2n\Delta^2}{\sigma^2}.
\end{equation}
Hence
\begin{equation}
    B_f=\frac{n\Delta^2}{\sigma^2}.
    \label{eq:appendix-gaussian-mean-B}
\end{equation}

\paragraph*{4) Apply the Bernoulli inverse}
For $s\le\Delta^2/2$, Step 2 gives $b_s(Q)\ge1/4$. Lemma~\ref{lem:reverse-kl}
therefore gives
\begin{align}
    a_f^-(B_f;b_s(Q))
    &\ge \frac1{16}e^{-4B_f}
    =\frac1{16}
    \exp\left(-\frac{4n\Delta^2}{\sigma^2}\right).
    \label{eq:appendix-gaussian-mean-a}
\end{align}

\paragraph*{5) Apply the CVaR bound}
Let $\delta:=1-\alpha$. Corollary~\ref{cor:general-cvar} gives
\begin{equation}
    \CVaR_\alpha^{\mu_\Delta}(\ell_\Delta)
    \ge
    \inf_{0\le s\le\Delta^2}
    \left\{
        s+\frac{\Delta^2}{\delta}a_f^-(B_f;b_s(Q))
    \right\}.
\end{equation}
If $s\ge\Delta^2/2$, the expression is at least $\Delta^2/2$. If
$0\le s<\Delta^2/2$, Step 4 gives
\begin{equation}
    s+\frac{\Delta^2}{\delta}a_f^-(B_f;b_s(Q))
    \ge
    \frac{\Delta^2}{16\delta}
    \exp\left(-\frac{4n\Delta^2}{\sigma^2}\right).
\end{equation}
Thus every $s\in[0,\Delta^2]$ satisfies the same lower bound, and therefore
\begin{equation}
    \CVaR_\alpha^{\mu_\Delta}(\ell_\Delta)
    \ge
    \Delta^2\min\left\{
        \frac12,
        \frac{1}{16(1-\alpha)}
        \exp\left(-\frac{4n\Delta^2}{\sigma^2}\right)
    \right\}.
\end{equation}
Combining this with \eqref{eq:gaussian-mean-proof-cvar-monotonicity} and
then taking the infimum over all estimators proves the fixed-prior bound
\eqref{eq:gaussian-mean-fixed} for the arbitrary $\Delta\in(0,R]$.

\paragraph*{6) Choose a hard prior from the family}
Assume $\delta=1-\alpha<1/8$ and set $x:=\Delta^2$. The fixed-prior bound is
\begin{equation}
    x\min\left\{
        \frac12,
        \frac{1}{16\delta}e^{-4nx/\sigma^2}
    \right\},
    \qquad 0<x\le R^2.
\end{equation}
Define
\begin{equation}
    g(x):=\frac{1}{16\delta}e^{-4nx/\sigma^2}.
\end{equation}
Since $g$ is decreasing, there is a unique $x_0>0$ satisfying
\begin{equation}
    g(x_0)=\frac12.
\end{equation}
Solving gives
\begin{equation}
    x_0=\frac{\sigma^2}{4n}\log\frac1{8\delta}.
\end{equation}
Choose
\begin{equation}
    x_{n,\alpha}:=\min\{R^2,x_0\},
    \qquad
    \Delta_{n,\alpha}:=\sqrt{x_{n,\alpha}}.
\end{equation}
This is an admissible separation, so the fixed-prior inequality applies to
$\mu_{\Delta_{n,\alpha}}$.

If $R^2\le x_0$, then $g(R^2)\ge1/2$, and the fixed-prior bound at
$x=R^2$ is $R^2/2$. If $R^2>x_0$, then $x=x_0$ and both terms inside the
minimum equal $1/2$, so the bound is $x_0/2$. Hence
\begin{equation}
    \inf_{\widehat\theta}
    \CVaR_\alpha^{\mu_{\Delta_{n,\alpha}}}
    \left((\widehat\theta-\theta)^2\right)
    \ge
    \frac12
    \min\left\{
        R^2,
        \frac{\sigma^2}{4n}
        \log\frac1{8(1-\alpha)}
    \right\}.
\end{equation}
Equivalently,
\begin{equation}
    \sup_{0<\Delta\le R}
    \inf_{\widehat\theta}
    \CVaR_\alpha^{\mu_\Delta}
    \left((\widehat\theta-\theta)^2\right)
    \ge
    \frac12
    \min\left\{
        R^2,
        \frac{\sigma^2}{4n}
        \log\frac1{8(1-\alpha)}
    \right\}.
\end{equation}
This is the stated hard-prior consequence.
\end{proof}

\subsection{K-armed Gaussian bandits}

\begin{proof}[Proof of Theorem~\ref{thm:gaussian-k-bandit}]
Fix $K\ge2$, $T\ge1$, $\sigma>0$, $\alpha\in(0,1)$, and an arbitrary
$g\in(0,1]$. Fix the uniform prior
\begin{equation}
    \mu_g^K=\frac1K\sum_{j=1}^K\delta_{M_j^g}
\end{equation}
and an arbitrary algorithm $\Alg\in\Dcal_T$. Let $M_0$ be the null model in
which every arm has reward distribution $\mathcal N(0,\sigma^2)$. Let $Q$
be the transcript distribution under $M_0$, and let $\mathbb P_j$ be the
transcript distribution under $M_j^g$.

\paragraph*{1) Define the loss}
Under $M_j^g$,
\begin{equation}
    R_T(M_j^g)=g(T-N_j(T)),
\end{equation}
so
\begin{equation}
    0\le R_T(M_j^g)\le Tg.
\end{equation}
Thus $L_{\max}=Tg$.

\paragraph*{2) Bound the reference excess loss}
For every realized action sequence,
\begin{equation}
    \sum_{j=1}^K N_j(T)=T.
\end{equation}
Hence
\begin{align}
    \frac1K\sum_{j=1}^K R_T(M_j^g)
    &
    =
    \frac{g}{K}\sum_{j=1}^K(T-N_j(T))
    \nonumber\\
    &=
    \frac{K-1}{K}gT.
    \label{eq:proof-kbandit-average-regret}
\end{align}
This is a pointwise identity and does not depend on $Q$.

For $s\in[0,Tg]$, convexity of $x\mapsto(x-s)_+$ gives
\begin{align}
    b_s(Q)
    &\ge
    \frac1{Tg}
    \left(
        \frac{K-1}{K}gT-s
    \right)_+
    \nonumber\\
    &=
    \left(
        \frac{K-1}{K}-\frac{s}{Tg}
    \right)_+.
\end{align}
Therefore,
\begin{equation}
    0\le s\le\frac{K-1}{2K}Tg
    \quad\Longrightarrow\quad
    b_s(Q)\ge\frac{K-1}{2K}\ge\frac14.
\end{equation}

\paragraph*{3) Decompose the adaptive transcript KL}
Let
\begin{equation}
    H_{t-1}:=(A_1,Y_1,\ldots,A_{t-1},Y_{t-1}).
\end{equation}
Conditional on $H_{t-1}=h$, the algorithm selects $A_t$ according to the same
kernel under $M_0$ and $M_j^g$. Thus the action component cancels in the
conditional KL. Given the selected arm, the reward distribution is independent
of the past. Under $M_0$,
\begin{equation}
    P_0(\cdot\mid a)=\mathcal N(0,\sigma^2),
    \qquad a\in[K],
\end{equation}
while under $M_j^g$,
\begin{equation}
    P_j(\cdot\mid a)
    =
    \begin{cases}
        \mathcal N(g,\sigma^2),&a=j,\\
        \mathcal N(0,\sigma^2),&a\ne j.
    \end{cases}
\end{equation}
The chain rule for KL therefore gives
\begin{align}
    \KL(Q\|\mathbb P_j)
    &
    =
    \sum_{t=1}^T
    \E_Q\left[
        \KL\left(
            P_0(\cdot\mid A_t)
            \middle\|
            P_j(\cdot\mid A_t)
        \right)
    \right].
\end{align}
The conditional KL is zero when $A_t\ne j$. When $A_t=j$,
\begin{equation}
    \KL\left(\mathcal N(0,\sigma^2)\middle\|\mathcal N(g,\sigma^2)\right)
    =\frac{g^2}{2\sigma^2}.
\end{equation}
Thus
\begin{align}
    \KL(Q\|\mathbb P_j)
    &=
    \frac{g^2}{2\sigma^2}
    \sum_{t=1}^T\E_Q[\1\{A_t=j\}]
    \\
    &=
    \frac{g^2}{2\sigma^2}\E_Q[N_j(T)].
    \label{eq:proof-kbandit-kl-j}
\end{align}

Averaging over the uniform prior gives
\begin{align}
    B_f
    &:={\frac{1}{K}}\sum_{j=1}^K\KL(Q\|\mathbb P_j)
    \nonumber\\
    &=
    \frac{g^2}{2K\sigma^2}
    \sum_{j=1}^K\E_Q[N_j(T)]
    \nonumber\\
    &=
    \frac{g^2T}{2K\sigma^2},
    \label{eq:appendix-k-B}
\end{align}
where the last equality uses
$\sum_jN_j(T)=T$ pathwise.

\paragraph*{4) Bernoulli inverse}
For $s\le (K-1)Tg/(2K)$, Lemma~\ref{lem:reverse-kl} and
$b_s(Q)\ge1/4$ give
\begin{equation}
    a_f^-(B_f;b_s(Q))
    \ge
    \frac1{16}
    \exp\left(-\frac{2g^2T}{K\sigma^2}\right).
\end{equation}

\paragraph*{5) CVaR bound}
Let $\delta:=1-\alpha$. If
$s\ge (K-1)Tg/(2K)$, the CVaR objective is at least that threshold. If
$s<(K-1)Tg/(2K)$, then
\begin{equation}
    s+\frac{Tg}{\delta}a_f^-(B_f;b_s(Q))
    \ge
    \frac{Tg}{16\delta}
    \exp\left(-\frac{2g^2T}{K\sigma^2}\right).
\end{equation}
Thus
\begin{equation}
    \CVaR_\alpha^{\mu_g^K}(R_T)
    \ge
    Tg
    \min\left\{
        \frac{K-1}{2K},
        \frac{1}{16(1-\alpha)}
        \exp\left(-\frac{2g^2T}{K\sigma^2}\right)
    \right\}.
\end{equation}
Since the algorithm was arbitrary, the fixed-prior inequality follows.

\paragraph*{6) Choose a hard prior from the gap family}
Set $\delta:=1-\alpha$ and
\begin{equation}
    c_K:=\frac{K-1}{2K}.
\end{equation}
The two terms inside the minimum are equal when
\begin{equation}
    c_K
    =
    \frac{1}{16\delta}
    \exp\left(-\frac{2g^2T}{K\sigma^2}\right).
\end{equation}
Thus
\begin{equation}
    g_0^2
    =
    \frac{K\sigma^2}{2T}
    \log\frac{K}{8(K-1)\delta}.
\end{equation}
When
$\delta<K/[8(K-1)]$, the logarithm is positive. Choose
\begin{equation}
    g_{T,\alpha}:=\min\{1,g_0\}.
\end{equation}
If $g_0\le1$, both branches equal $c_K$. If $g_0>1$, monotonicity of the
exponential branch implies that at $g=1$ the second branch is at least $c_K$.
Hence, in either case,
\begin{equation}
    \inf_{\Alg\in\Dcal_T}
    \CVaR_\alpha^{\mu_{g_{T,\alpha}}^K}(R_T)
    \ge
    c_KT
    \min\left\{
        1,
        \sigma\sqrt{\frac{K}{2T}}
        \sqrt{\log\frac{K}{8(K-1)(1-\alpha)}}
    \right\}.
\end{equation}
For fixed $K$, this is of order
\begin{equation}
    \min\left\{
        T,
        \sigma\sqrt{KT\log\frac1{1-\alpha}}
    \right\}.
\end{equation}
\end{proof}

%% file: appendices/D_ldp.tex
\section{Proofs for the LDP Applications}

This appendix proves the transcript-divergence result and the LDP
applications. We first derive the adaptive KL bound for released
transcripts and then apply it to Gaussian mean estimation, Gaussian
$K$-armed bandits, and the bounded mean-estimation comparison.
\label{app:ldp}

\begin{proof}[Proof of Proposition~\ref{prop:ldp-transcript}]
Fix a horizon $T$, two models $M$ and $M'$, and an arbitrary
$\epsvec$-LDP procedure $\Alg$. We prove the transcript divergence bound
by conditioning on the released history at each round.

Let
\begin{equation}
    H_{t-1}^{\mathrm{LDP}}
    :=
    (A_1,Z_1,\ldots,A_{t-1},Z_{t-1})
\end{equation}
be the released history before round $t$. Conditional on
$H_{t-1}^{\mathrm{LDP}}=h$, the algorithm selects $A_t$ according to
an action kernel
\begin{equation}
    \pi_t(da\mid h).
\end{equation}
The same kernel is used under $M$ and $M'$, because the algorithm is the
same in the two experiments.

Given $(A_t,H_{t-1}^{\mathrm{LDP}})=(a,h)$, model $M$ generates the raw
observation according to
\begin{equation}
    P_t^M(dy\mid a,h),
\end{equation}
and the LDP channel produces the released observation according to
\begin{equation}
    W_t(dz\mid y,a,h).
\end{equation}
Therefore the conditional released marginal is
\begin{equation}
    \widetilde P_t^M(dz\mid a,h)
    =
    \int W_t(dz\mid y,a,h)
    P_t^M(dy\mid a,h).
\end{equation}
Define $\widetilde P_t^{M'}$ analogously.

Conditional on $H_{t-1}^{\mathrm{LDP}}=h$, the joint distributions of
$(A_t,Z_t)$ factor as
\begin{align}
    \mathbb P_{\mathrm{LDP}}^{M,\Alg}(da,dz\mid h)
    &
    =
    \pi_t(da\mid h)\widetilde P_t^M(dz\mid a,h),
    \\
    \mathbb P_{\mathrm{LDP}}^{M',\Alg}(da,dz\mid h)
    &
    =
    \pi_t(da\mid h)\widetilde P_t^{M'}(dz\mid a,h).
\end{align}
Since the action kernel is identical, it contributes zero to the
conditional KL divergence. The KL chain rule therefore gives
\begin{align}
    &\KL\left(
        \mathbb P_{\mathrm{LDP}}^{M,\Alg}
        \middle\|
        \mathbb P_{\mathrm{LDP}}^{M',\Alg}
    \right)
    \\[-1mm]
    &\qquad=
    \sum_{t=1}^T
    \E_{\mathbb P_{\mathrm{LDP}}^{M,\Alg}}
    \left[
        \KL\left(
            \widetilde P_t^M(\cdot\mid A_t,H_{t-1}^{\mathrm{LDP}})
            \middle\|
            \widetilde P_t^{M'}(\cdot\mid A_t,H_{t-1}^{\mathrm{LDP}})
        \right)
    \right].
    \label{eq:ldp-proof-chain-rule}
\end{align}

For fixed $t$, $h$, and $a$, the one-direction KL is bounded by the
symmetrized KL. Hence, by \eqref{eq:conditional-ldp-contraction},
\begin{equation}
    \KL\left(
        \widetilde P_t^M(\cdot\mid a,h)
        \middle\|
        \widetilde P_t^{M'}(\cdot\mid a,h)
    \right)
    \le
    \kappa(\varepsilon_t)
    d_{\mathrm{TV}}^2
    \left(
        P_t^M(\cdot\mid a,h),
        P_t^{M'}(\cdot\mid a,h)
    \right).
    \label{eq:ldp-proof-one-direction}
\end{equation}
Substituting \eqref{eq:ldp-proof-one-direction} into
\eqref{eq:ldp-proof-chain-rule} proves
\begin{equation}
    \KL\left(
        \mathbb P_{\mathrm{LDP}}^{M,\Alg}
        \middle\|
        \mathbb P_{\mathrm{LDP}}^{M',\Alg}
    \right)
    \le
    \E_{\mathbb P_{\mathrm{LDP}}^{M,\Alg}}
    \left[
        \sum_{t=1}^T
        \kappa(\varepsilon_t)
        d_t(H_{t-1}^{\mathrm{LDP}},A_t)^2
    \right].
\end{equation}
which is \eqref{eq:adaptive-ldp-kl}.
\end{proof}

\subsection{Gaussian mean estimation under heterogeneous LDP}

\begin{proof}[Proof of Theorem~\ref{thm:gaussian-ldp-mean}]
Fix $n\ge1$, $\sigma>0$, $\alpha\in(0,1)$, and an arbitrary
$\Delta\in(0,R]$. We first prove the fixed-prior statement for
\begin{equation}
    \mu_\Delta
    =
    \frac12\delta_{-\Delta}
    +
    \frac12\delta_{+\Delta}.
\end{equation}
Fix an arbitrary procedure
$\Alg\in\mathcal D_{\boldsymbol\varepsilon}^{\mathrm{LDP}}$.
The choice of $\Delta$ will be made only after the fixed-prior bound has
been established.

\paragraph*{1) Define a bounded loss}
Define
\begin{equation}
    \ell_\Delta(\theta,\widehat\theta)
    :=
    \min\left\{
        (\widehat\theta-\theta)^2,
        \Delta^2
    \right\}.
\end{equation}
Then $0\le\ell_\Delta\le\Delta^2$, so
$L_{\max}=\Delta^2$. Moreover,
\begin{equation}
    \ell_\Delta(\theta,\widehat\theta)
    \le
    (\widehat\theta-\theta)^2.
\end{equation}
Thus CVaR monotonicity gives
\begin{equation}
    \CVaR_\alpha^{\mu_\Delta,\Alg}
    \left((\widehat\theta-\theta)^2\right)
    \ge
    \CVaR_\alpha^{\mu_\Delta,\Alg}(\ell_\Delta).
    \label{eq:ldp-mean-monotonicity}
\end{equation}

\paragraph*{2) Lower bound the reference excess loss}
Let $Q$ be the released transcript distribution when the true parameter
is $-
\Delta$. For $z\in\R$, define
\begin{equation}
    \ell_+(z):=\min\{(z-\Delta)^2,\Delta^2\},
    \qquad
    \ell_-(z):=\min\{(z+\Delta)^2,\Delta^2\}.
\end{equation}
Because the two parameter values are separated by $2\Delta$,
\begin{equation}
    \max\{\ell_+(z),\ell_-(z)\}=\Delta^2
\end{equation}
for every $z$. Consequently, for every
$s\in[0,\Delta^2]$,
\begin{equation}
    (\ell_+(z)-s)_++(\ell_-(z)-s)_+
    \ge\Delta^2-s.
    \label{eq:ldp-mean-separation}
\end{equation}

The fixed prior is uniform, so
\begin{align}
    b_s(Q)
    &
    =
    \frac{1}{2\Delta^2}
    \E_Q
    \left[
        (\ell_+(\widehat\theta)-s)_+
        +(\ell_-(\widehat\theta)-s)_+
    \right]
    \nonumber\\
    &\ge
    \frac{\Delta^2-s}{2\Delta^2}
    =
    \frac12\left(1-\frac{s}{\Delta^2}\right).
\end{align}
Hence
\begin{equation}
    0\le s\le\frac{\Delta^2}{2}
    \quad\Longrightarrow\quad
    b_s(Q)\ge\frac14.
    \label{eq:ldp-mean-bs-quarter}
\end{equation}

\paragraph*{3) Compute the total variation distance of the raw laws}
For every user $i$, the two raw distributions are
\begin{equation}
    P_{+,\Delta}=\mathcal N(+\Delta,\sigma^2),
    \qquad
    P_{-,\Delta}=\mathcal N(-\Delta,\sigma^2).
\end{equation}
Their KL divergence in either direction is
\begin{equation}
    \KL(P_{+,\Delta}\|P_{-,\Delta})
    =
    \KL(P_{-,\Delta}\|P_{+,\Delta})
    =
    \frac{2\Delta^2}{\sigma^2}.
\end{equation}
Pinsker's inequality therefore gives
\begin{equation}
    \TV(P_{+,\Delta},P_{-,\Delta})^2
    \le
    \frac{\Delta^2}{\sigma^2}.
    \label{eq:ldp-mean-tv}
\end{equation}

\paragraph*{4) Bound the released transcript divergence}
Let $\widetilde P_+$ and $\widetilde P_-$ denote the released transcript
laws under $+\Delta$ and $-\Delta$, respectively. Since the problem is
passive, Proposition~\ref{prop:ldp-transcript} reduces to the sum of the
one-user conditional divergences. By \eqref{eq:ldp-mean-tv},
\begin{align}
    \KL(\widetilde P_-\|\widetilde P_+)
    &\le
    \frac{\Delta^2}{\sigma^2}
    \sum_{i=1}^n\kappa(\varepsilon_i).
\end{align}
Define
\begin{equation}
    \Lambda_n:=\sum_{i=1}^n\kappa(\varepsilon_i).
\end{equation}
For reverse KL,
\begin{equation}
    B_f
    =
    \frac12\KL(\widetilde P_-\|\widetilde P_-)
    +
    \frac12\KL(\widetilde P_-\|\widetilde P_+)
    =
    \frac12\KL(\widetilde P_-\|\widetilde P_+).
\end{equation}
Thus
\begin{equation}
    B_f
    \le
    \frac{\Delta^2\Lambda_n}{2\sigma^2}.
    \label{eq:appendix-ldp-mean-B}
\end{equation}
If $\varepsilon_i\le1$ for all $i$, then
\begin{equation}
    \Lambda_n
    \le
    16\sum_{i=1}^n\varepsilon_i^2,
\end{equation}
so
\begin{equation}
    B_f
    \le
    \frac{8\Delta^2}{\sigma^2}
    \sum_{i=1}^n\varepsilon_i^2.
    \label{eq:appendix-ldp-mean-B-simple}
\end{equation}

\paragraph*{5) Bound the Bernoulli inverse}
For $s\le\Delta^2/2$, 
\eqref{eq:ldp-mean-bs-quarter} gives $b_s(Q)\ge1/4$. Hence Lemma~\ref{lem:reverse-kl} gives
\begin{align}
    a_f^-(B_f;b_s(Q))
    &\ge
    \frac1{16}e^{-4B_f}
    \nonumber\\
    &\ge
    \frac1{16}
    \exp\left(-\frac{2\Delta^2\Lambda_n}{\sigma^2}\right).
    \label{eq:ldp-mean-inverse}
\end{align}

\paragraph*{6) Obtain the fixed-prior CVaR bound}
Let $\delta:=1-\alpha$. Corollary~\ref{cor:general-cvar} gives
\begin{equation}
    \CVaR_\alpha^{\mu_\Delta,\Alg}(\ell_\Delta)
    \ge
    \inf_{0\le s\le\Delta^2}
    \left\{
        s+\frac{\Delta^2}{\delta}
        a_f^-(B_f;b_s(Q))
    \right\}.
\end{equation}
If $s\ge\Delta^2/2$, the expression is at least $\Delta^2/2$. If
$0\le s<\Delta^2/2$, then \eqref{eq:ldp-mean-inverse} and $s\ge0$ give
\begin{equation}
    s+\frac{\Delta^2}{\delta}a_f^-(B_f;b_s(Q))
    \ge
    \frac{\Delta^2}{16\delta}
    \exp\left(-\frac{2\Delta^2\Lambda_n}{\sigma^2}\right).
\end{equation}
Therefore, for every $s\in[0,\Delta^2]$,
\begin{equation}
    s+\frac{\Delta^2}{\delta}a_f^-(B_f;b_s(Q))
    \ge
    \Delta^2
    \min\left\{
        \frac12,
        \frac{1}{16(1-\alpha)}
        \exp\left(-\frac{2\Delta^2\Lambda_n}{\sigma^2}\right)
    \right\}.
\end{equation}
Taking the infimum over $s$, using \eqref{eq:ldp-mean-monotonicity}, and then
using the arbitrariness of $\Alg$ gives the fixed-prior bound
\eqref{eq:gaussian-ldp-mean-fixed} for every $\Delta\in(0,R]$.

\paragraph*{7) Choose a hard prior from the family}
Assume $\delta<1/8$. The two terms in the minimum are equal when
\begin{equation}
    \frac12
    =
    \frac{1}{16\delta}
    \exp\left(-\frac{2\Delta^2\Lambda_n}{\sigma^2}\right).
\end{equation}
Hence
\begin{equation}
    \Delta_0^2
    =
    \frac{\sigma^2}{2\Lambda_n}
    \log\frac{1}{8\delta}.
\end{equation}
Choose
\begin{equation}
    \Delta_{n,\alpha}^2
    :=
    \min\{R^2,\Delta_0^2\}.
\end{equation}
Because the fixed-prior inequality holds for every admissible $\Delta$,
it applies to this choice. If $R^2\le\Delta_0^2$, the exponential term
is at least $1/2$ and the lower bound is $R^2/2$. If $R^2>\Delta_0^2$,
the two terms are equal to $1/2$ and the lower bound is
$\Delta_0^2/2$. Therefore
\begin{equation}
    \inf_{\Alg\in\mathcal D_{\boldsymbol\varepsilon}^{\mathrm{LDP}}}
    \CVaR_\alpha^{\mu_{\Delta_{n,\alpha}},\Alg}
    \left((\widehat\theta-\theta)^2\right)
    \ge
    \frac12
    \min\left\{
        R^2,
        \frac{\sigma^2}{2\Lambda_n}
        \log\frac{1}{8(1-\alpha)}
    \right\}.
\end{equation}
\end{proof}

\subsection{Gaussian K-armed bandits under LDP}

\begin{proof}[Proof of Theorem~\ref{thm:gaussian-ldp-bandit}]
Fix $K\ge2$, $T\ge1$, $\sigma>0$, $\alpha\in(0,1)$, and an arbitrary
gap $g\in(0,1]$. Fix the prior
\begin{equation}
    \mu_g^K=\frac1K\sum_{j=1}^K\delta_{M_j^g}
\end{equation}
and an arbitrary procedure
\begin{equation}
    \Alg\in\mathcal D_{T,\boldsymbol\varepsilon}^{\mathrm{LDP}}.
\end{equation}
Let $M_0$ be the null model in which every arm has reward law
$\mathcal N(0,\sigma^2)$. Let $Q$ denote its released transcript law,
and let $\mathbb P_j$ denote the released transcript law under $M_j^g$.

\paragraph*{1) Define the loss}
Under $M_j^g$,
\begin{equation}
    R_T(M_j^g)=g(T-N_j(T)),
\end{equation}
so
\begin{equation}
    0\le R_T(M_j^g)\le Tg.
\end{equation}
Thus $L_{\max}=Tg$.

\paragraph*{2) Lower bound the reference excess loss}
For every realized action sequence,
\begin{equation}
    \sum_{j=1}^K N_j(T)=T.
\end{equation}
Hence
\begin{align}
    \frac1K\sum_{j=1}^K R_T(M_j^g)
    &
    =
    \frac{K-1}{K}gT.
\end{align}
This identity is pointwise and therefore does not depend on the reference
transcript law $Q$.

For $s\in[0,Tg]$, Jensen's inequality gives
\begin{align}
    b_s(Q)
    &\ge
    \frac1{Tg}
    \left(
        \frac{K-1}{K}gT-s
    \right)_+\nonumber\\
    &=
    \left(
        \frac{K-1}{K}-\frac{s}{Tg}
    \right)_+.
\end{align}
Therefore
\begin{equation}
    0\le s\le\frac{K-1}{2K}Tg
    \quad\Longrightarrow\quad
    b_s(Q)\ge\frac{K-1}{2K}\ge\frac14.
\end{equation}

\paragraph*{3) Bound the raw Gaussian total variation distance}
Fix an arm $j$. Under $M_0$ and $M_j^g$, the reward distributions coincide
for every arm $a\ne j$. When $a=j$, they are
\begin{equation}
    \mathcal N(0,\sigma^2)
    \qquad\text{and}\qquad
    \mathcal N(g,\sigma^2).
\end{equation}
Pinsker's inequality and the equal-variance Gaussian KL formula give
\begin{align}
    \TV^2
    \left(
        \mathcal N(0,\sigma^2),
        \mathcal N(g,\sigma^2)
    \right)
    &\le
    \frac12
    \KL\left(
        \mathcal N(0,\sigma^2)
        \middle\|
        \mathcal N(g,\sigma^2)
    \right)
    \nonumber\\
    &=
    \frac{g^2}{4\sigma^2}.
    \label{eq:gaussian-ldp-bandit-tv-proof}
\end{align}

\paragraph*{4) Bound the released transcript divergence}
Fix $j\in[K]$ and let
\begin{equation}
    H_{t-1}^{\mathrm{LDP}}
    :=
    (A_1,Z_1,\ldots,A_{t-1},Z_{t-1}).
\end{equation}
The conditional raw reward laws under $M_0$ and $M_j^g$ differ only when
$A_t=j$. Therefore Proposition~\ref{prop:ldp-transcript} and
\eqref{eq:gaussian-ldp-bandit-tv-proof} give
\begin{align}
    \KL(Q\|\mathbb P_j)
    &\le
    \frac{g^2}{4\sigma^2}
    \E_Q
    \left[
        \sum_{t=1}^T
        \kappa(\varepsilon_t)\1\{A_t=j\}
    \right].
    \label{eq:gaussian-ldp-bandit-kl-j}
\end{align}

For reverse KL,
\begin{align}
    B_f
    &:=
    \frac1K\sum_{j=1}^K\KL(Q\|\mathbb P_j)
    \nonumber\\
    &\le
    \frac{g^2}{4K\sigma^2}
    \E_Q
    \left[
        \sum_{t=1}^T
        \kappa(\varepsilon_t)
        \sum_{j=1}^K\1\{A_t=j\}
    \right]
    \nonumber\\
    &=
    \frac{g^2}{4K\sigma^2}
    \sum_{t=1}^T\kappa(\varepsilon_t).
    \label{eq:appendix-ldp-bandit-B}
\end{align}
Define
\begin{equation}
    \Lambda_T
    :=
    \sum_{t=1}^T\kappa(\varepsilon_t).
\end{equation}
Then
\begin{equation}
    B_f
    \le
    \frac{g^2\Lambda_T}{4K\sigma^2}.
    \label{eq:gaussian-ldp-bandit-B-final}
\end{equation}
If $\varepsilon_t\le1$ for every $t$, then
\begin{equation}
    \Lambda_T
    \le
    16\sum_{t=1}^T\varepsilon_t^2,
\end{equation}
so
\begin{equation}
    B_f
    \le
    \frac{4g^2}{K\sigma^2}
    \sum_{t=1}^T\varepsilon_t^2.
\end{equation}

\paragraph*{5) Bound the Bernoulli inverse}
For
\begin{equation}
    0\le s\le\frac{K-1}{2K}Tg,
\end{equation}
we have $b_s(Q)\ge1/4$. Hence Lemma~\ref{lem:reverse-kl} gives
\begin{align}
    a_f^-(B_f;b_s(Q))
    &\ge
    \frac1{16}e^{-4B_f}
    \nonumber\\
    &\ge
    \frac1{16}
    \exp\left(-\frac{g^2\Lambda_T}{K\sigma^2}\right).
    \label{eq:gaussian-ldp-bandit-inverse-proof}
\end{align}

\paragraph*{6) Obtain the fixed-prior CVaR bound}
Set $\delta:=1-\alpha$. Corollary~\ref{cor:general-cvar} gives
\begin{equation}
    \CVaR_\alpha^{\mu_g^K,\Alg}(R_T)
    \ge
    \inf_{0\le s\le Tg}
    \left\{
        s+\frac{Tg}{\delta}
        a_f^-(B_f;b_s(Q))
    \right\}.
\end{equation}
If
\begin{equation}
    s\ge\frac{K-1}{2K}Tg,
\end{equation}
then the expression is at least $(K-1)Tg/(2K)$. If
$0\le s<(K-1)Tg/(2K)$, then \eqref{eq:gaussian-ldp-bandit-inverse-proof}
gives
\begin{equation}
    s+\frac{Tg}{\delta}a_f^-(B_f;b_s(Q))
    \ge
    \frac{Tg}{16\delta}
    \exp\left(-\frac{g^2\Lambda_T}{K\sigma^2}\right).
\end{equation}
Therefore
\begin{equation}
    \CVaR_\alpha^{\mu_g^K,\Alg}(R_T)
    \ge
    Tg
    \min\left\{
        \frac{K-1}{2K},
        \frac{1}{16(1-\alpha)}
        \exp\left(-\frac{g^2\Lambda_T}{K\sigma^2}\right)
    \right\}.
\end{equation}
Since $\Alg$ was arbitrary, this proves the fixed-prior bound
\eqref{eq:gaussian-ldp-bandit-fixed}.

\paragraph*{7) Choose a hard prior from the gap family}
Let
\begin{equation}
    c_K:=\frac{K-1}{2K},
    \qquad
    \delta:=1-\alpha.
\end{equation}
Assume
\begin{equation}
    \delta<\frac{K}{8(K-1)}.
\end{equation}
The two terms in the minimum in \eqref{eq:gaussian-ldp-bandit-fixed} are equal
when
\begin{equation}
    c_K
    =
    \frac{1}{16\delta}
    \exp\left(-\frac{g^2\Lambda_T}{K\sigma^2}\right).
\end{equation}
Solving for $g^2$ gives
\begin{equation}
    g_0^2
    =
    \frac{K\sigma^2}{\Lambda_T}
    \log\frac{K}{8(K-1)\delta}.
\end{equation}
Choose
\begin{equation}
    g_{T,\alpha}:=\min\{1,g_0\}.
\end{equation}
If $g_0\le1$, the two branches are equal at $g=g_0$. If $g_0>1$, then
monotonicity of the exponential branch implies that at $g=1$ the second
branch is at least $c_K$. Hence, in both cases,
\begin{equation}
\begin{aligned}
    &\inf_{\Alg\in\mathcal D_{T,\boldsymbol\varepsilon}^{\mathrm{LDP}}}
    \CVaR_\alpha^{\mu_{g_{T,\alpha}}^K,\Alg}(R_T)\\
    &\qquad\ge
    \frac{K-1}{2K}T
    \min\left\{
        1,
        \frac{\sigma\sqrt K}{\sqrt{\Lambda_T}}
        \sqrt{\log\frac{K}{8(K-1)(1-\alpha)}}
    \right\}.
\end{aligned}
\end{equation}
If $0<\varepsilon_t=\varepsilon\le1$ for every $t$, then
$\Lambda_T\le16T\varepsilon^2$ and therefore, for fixed $K$,
\begin{equation}
    \inf_{\Alg\in\mathcal D_{T,\varepsilon}^{\mathrm{LDP}}}
    \CVaR_\alpha^{\mu_{g_{T,\alpha}}^K,\Alg}(R_T)
    \gtrsim
    \min\left\{
        T,
        \frac{\sigma}{\varepsilon}
        \sqrt{KT\log\frac1{1-\alpha}}
    \right\},
\end{equation}
up to constants depending only on $K$.
\end{proof}

\subsection{Bounded mean estimation for comparison with heterogeneous LDP}

\begin{proof}[Proof of Theorem~\ref{thm:bounded-ldp-mean}]
Fix $n\ge1$, $\alpha\in(0,1)$, and an arbitrary
\begin{equation}
    \Delta\in(0,1].
\end{equation}
Consider the two bounded distributions
\begin{align}
    P_{+,\Delta}(X=1)&=\frac{1+\Delta}{2},
    &P_{+,\Delta}(X=-1)&=\frac{1-\Delta}{2},\\
    P_{-,\Delta}(X=1)&=\frac{1-\Delta}{2},
    &P_{-,\Delta}(X=-1)&=\frac{1+\Delta}{2}.
\end{align}
Their means are $+\Delta$ and $-\Delta$. Fix the prior
\begin{equation}
    \mu_\Delta^{\mathrm{bd}}
    =
    \frac12\delta_{P_{+,\Delta}}
    +
    \frac12\delta_{P_{-,\Delta}}.
\end{equation}
Fix an arbitrary heterogeneous LDP procedure.

\paragraph*{1) Bounded loss and reference excess loss}
Use
\begin{equation}
    \ell_\Delta
    =
    \min\left\{
        (\widehat\theta-\theta)^2,
        \Delta^2
    \right\}.
\end{equation}
Let $Q$ denote the released transcript distribution under
$P_{-,\Delta}$.
The same pointwise separation as in the Gaussian mean-estimation proof gives
\begin{equation}
    b_s(Q)
    \ge
    \frac12
    \left(1-\frac{s}{\Delta^2}\right),
\end{equation}
so $b_s(Q)\ge1/4$ for
$0\le s\le\Delta^2/2$.

\paragraph*{2) Compute the raw total variation distance}
Since the distributions are supported on $\{-1,+1\}$,
\begin{align}
    \TV(P_{+,\Delta},P_{-,\Delta})
    &
    =
    \frac12
    \sum_{x\in\{-1,+1\}}
    \left|P_{+,\Delta}(x)-P_{-,\Delta}(x)\right|
    \nonumber\\
    &=
    \frac12(\Delta+\Delta)
    =\Delta.
    \label{eq:bounded-mean-tv-proof}
\end{align}

\paragraph*{3) Bound the released divergence}
For each user $i$, the LDP contraction inequality gives
\begin{equation}
    \KL(\widetilde P_{-,i}\|\widetilde P_{+,i})
    \le
    \kappa(\varepsilon_i)\Delta^2.
\end{equation}
Applying the chain rule across the independent users and then using the
reverse-KL definition of $B_f$ gives
\begin{equation}
    B_f
    \le
    \frac{\Delta^2}{2}
    \sum_{i=1}^n\kappa(\varepsilon_i).
    \label{eq:bounded-mean-B-proof}
\end{equation}
When $\varepsilon_i\le1$,
\begin{equation}
    B_f
    \le
    8\Delta^2\sum_{i=1}^n\varepsilon_i^2.
\end{equation}

\paragraph*{4) Apply the Bernoulli inverse and the CVaR bound}
For $s\le\Delta^2/2$, Lemma~\ref{lem:reverse-kl} gives
\begin{equation}
    a_f^-(B_f;b_s(Q))
    \ge
    \frac1{16}
    \exp\left(-32\Delta^2\sum_{i=1}^n\varepsilon_i^2\right).
\end{equation}
The same two-range argument as above therefore gives
\begin{equation}
    \inf_{\Alg\in\mathcal D_{\boldsymbol\varepsilon}^{\mathrm{LDP}}}
    \CVaR_\alpha^{\mu_\Delta^{\mathrm{bd}},\Alg}
    \left((\widehat\theta-\theta)^2\right)
    \ge
    \Delta^2
    \min\left\{
        \frac12,
        \frac{1}{16(1-\alpha)}
        \exp\left(-32\Delta^2\sum_{i=1}^n\varepsilon_i^2\right)
    \right\}.
\end{equation}
Finally, choosing
\begin{equation}
    \Delta^2
    =
    \min\left\{
        1,
        \frac{1}{32\sum_{i=1}^n\varepsilon_i^2}
        \log\frac1{8(1-\alpha)}
    \right\}
\end{equation}
for $1-\alpha<1/8$ gives the stated comparison result.
\end{proof}